\documentclass[accepted, mathfont=ptmx]{uai2024} 
                        
\usepackage[american]{babel}

\usepackage{mathtools} 
\usepackage{booktabs} 

\def\ci{\perp\!\!\!\perp}

\usepackage[T1]{fontenc}
\usepackage[utf8]{inputenc}
\usepackage{authblk}
\usepackage{amssymb}
\usepackage{amsthm}
\usepackage{tikz}
\usepackage{caption}
\usepackage{subcaption}
\usepackage{graphicx}
\usepackage{algorithm}
\usepackage{algpseudocode}
\usepackage{algorithmicx}
\usepackage{listings}

\usepackage{natbib} 
\usepackage{hyperref}
\hypersetup{colorlinks=true,citecolor=black}

\newtheorem{theorem}{Theorem}
\newtheorem{proposition}{Proposition}

\newtheorem{lemma}{Lemma}

\newtheorem*{remark}{Remark}

\newenvironment{thma}[1]{\par\noindent{\bf Theorem #1.\ }\em}{\em}
\newenvironment{lma}[1]{\par\noindent{\bf Lemma #1.\ }\em}{\em}
\newenvironment{propa}[1]{\par\noindent{\bf Proposition #1.\ }\em}{\em}

\usepackage[normalem]{ulem}
\newcommand{\stkout}[1]{\ifmmode\text{\sout{\ensuremath{#1}}}\else\sout{#1}\fi}

\DeclareMathOperator*{\Perp}{\perp\!\!\!\perp}

\begin{document}
\title{
Zero Inflation as a Missing Data Problem: a Proxy-based Approach
}


\author[1]{\vspace{-2em}Trung Phung\thanks{tphung1@jhu.edu}}
\author[1]{Jaron J.R. Lee} 
\author[2]{Opeyemi Oladapo-Shittu} 
\author[2]{Eili Y. Klein} 
\author[1,2,3,4]{Ayse Pinar Gurses} 
\author[3]{Susan M. Hannum} 
\author[5,6]{Kimberly Weems} 
\author[3,4]{Jill A. Marsteller} 
\author[2,4,5]{Sara E. Cosgrove} 
\author[2,4]{Sara C. Keller} 
\author[1]{Ilya Shpitser} 

\affil[1]{Johns Hopkins Whiting School of Engineering, Baltimore, MD}
\affil[2]{Johns Hopkins University School of Medicine, Baltimore, MD}
\affil[3]{Johns Hopkins Bloomberg School of Public Health, Baltimore, MD}
\affil[4]{Johns Hopkins Medicine, Baltimore, MD}
\affil[5]{Johns Hopkins Health System, Baltimore, MD}
\affil[6]{Vassar Brothers Medical Center, Poughkeepsie, NY}

\maketitle

\begin{abstract}
\vspace{-2em}
A common type of zero-inflated data has certain true values incorrectly replaced by zeros due to data recording conventions (rare outcomes assumed to be absent) or details of data recording equipment (e.g. artificial zeros in gene expression data).

Existing methods for zero-inflated data either fit the observed data likelihood via parametric mixture models that explicitly represent excess zeros, or aim to replace excess zeros by imputed values.  If the goal of the analysis relies on knowing true data realizations, a particular challenge with zero-inflated data is identifiability, since it is difficult to correctly determine which observed zeros are real and which are inflated.

This paper views zero-inflated data as a general type of missing data problem, where the observability indicator for a potentially censored variable is itself unobserved whenever a zero is recorded. We show that, without additional assumptions, target parameters involving a zero-inflated variable are not identified. However, if a proxy of the missingness indicator is observed, a modification of the effect restoration approach of Kuroki and Pearl allows identification and estimation, given the proxy-indicator relationship is known.

If this relationship is unknown, our approach yields a partial identification strategy for sensitivity analysis. Specifically, we show that only certain proxy-indicator relationships are compatible with the observed data distribution. We give an analytic bound for this relationship in cases with a categorical outcome, which is sharp in certain models.
For more complex cases, sharp numerical bounds may be computed using methods in \citet{duarte23automated}.

We illustrate our method via simulation studies and a data application on central line-associated bloodstream infections (CLABSIs).
\vspace{-2em}
\end{abstract}

\section{Introduction}

Zero-inflated (ZI) data is prevalent in many empirical sciences such as
public health, epidemiology, computational biology, and medical research.
An important type of zero inflation occurs when some observed zeros of an outcome of interest do not represent true zero values.

As an example, consider patient surveillance for complications in outpatient settings, where any complication developed outside the hospital is of interest.
One such complication is a central line-associated bloodstream infection (CLABSI) which can occur in patients undergoing therapies involving central venous catheters (CVCs).
Such complications are fairly rare, but are associated with significant morbidity and mortality, and their prevalence is often assessed retrospectively.  Because of this, absence of sufficient information on whether such a complication is present in a particular patient is often coded as a ``presumed negative'' rather than a ``missing value'' \citep{Keller2020ReachingCLABSI}.
Since this type of value differs from a true negative value, indicating actual absence of a complication in a patient, the result is zero-inflated data.
Another prominent example is single-cell RNA sequence data, whose zeros may signify either genuine values (representing, e.g. lack of gene expression) 
or artificial zeros resulting from technical artifacts of experimental protocols or recording equipment \citep{Wagner2016Reveal, Jiang2022scRNAseqdata}.
In all these cases, naive analysis of ZI data that does not distinguish true from artificial zeros can lead to markedly biased conclusions.

Existing approaches for zero inflation focus on observed data likelihood modelling using either hurdle models or zero-inflation models
\citep{neelon2016modelingzero, Greene2005CensoredData}.
Hurdles models are mixtures models
of a distribution truncated at zero
and another distribution modeling the occurrence of $0$ values \citep{mullahy1986specification}.
In genomics applications, \citet{Yu.Drton.ea:2023:DirectedGraphical, Dai.Ng.ea:2023:GeneRegulatory} use graphical models to represent the zero-inflated likelihood for the purposes of causal discovery.
On the other hand, zero inflation models \citep{Lambert1992ZIP,Young2022ZeroinflatedmodelI} assume two sources of zeros, either structural (or inflated) zeros or true zeros due to sampling.
More recent work has extended this type of approach to include
semi-parametric models \citep{arab2012semiparametric, lam2006semiparametric}.
\citet{kleinke2013multipleimputation} apply an augmentation of the chained equations imputation approach to correct the bias introduced by inflated zeros.
\citet{Lukusa.Lee.ea:2017:ReviewZeroInflated} review
methods in settings where inflated zeros co-occur with missing data, however these settings do not include cases considered here, where the excess zeros represent a censored realization.

The disadvantage of the first type of approach is that it does not aim to reconstruct underlying values, which are often of interest.
The disadvantage of the second type of approach is that correctly distinguishing true from inflated zero values relies on assumptions that are unlikely to hold in practice, e.g., strict parametric assumptions. Moreover, these assumptions may not be congenial and not lead to a coherent full data distribution – guaranteeing model misspecification.
This is a more general issue than zero inflation, and occurs in standard missing data problems as well.
In contrast, our approach to modeling inflated zeros has two important features. First, we aim to distinguish true from inflated zeros, and thus identify underlying realizations in the data. Second, we avoid imposing strong parametric assumptions to do so.

Specifically, we propose to model zero inflation using a generalization of missing data models.
In standard missing data, the relationship between an observed variable and its corresponding underlying variable is determined by an \emph{observability indicator}.
If the indicator is $1$, the observed and the underlying variables coincide, while if the indicator is $0$, the observed variable is recorded as a missing value.
In zero inflated problems, we view improperly recorded zero values as missing values denoted by a zero.
Hence, in this view, we cannot tell a zero indicating an actual value from a zero indicating missingness,
and observing a zero means the observability indicator is \emph{itself} unobserved.

This complication implies that even if we assume a missing data model where the full data distribution would have been identified absent zero inflation,
such as the Missing-Completely-At-Random (MCAR) model, we would generally not obtain identification in the presence of zero inflation.
Thus, the variant of the missing data problem we consider is significantly more complicated than standard missing data.

We approach this problem using recent theory of graphical models applied to missing data, which gives general identification results in the absence of zero inflation
\citep{mohan13missing, bhattacharya19mid, bhattacharya20completeness}.
We first note that
zero inflation problems viewed in this framework could be 
arranged in a hierarchy similarly to missing data problems
\citep{rubin76inference}:
Zero-Inflated Missing-Completely-At-Random (ZI MCAR),
Zero-Inflated Missing-At-Random (ZI MAR), and
Zero-Inflated Missing-Not-At-Random (ZI MNAR).

We then show that if zero inflation is present, target parameters involving zero inflated variables are not identified without additional assumptions, even in the relatively simple ZI MCAR model.
We further show that if an informative proxy for a missingness indicator exists, identification of the target parameters becomes possible provided the missing data model (sans zero inflation) is identified, via a modification of the effect restoration approach in \citet{kuroki14measurement}, provided the true proxy-indicator relationship is known.

If this relationship is not known, we show that only certain proxy-indicator relationships are compatible with the overall model which provides a natural sensitivity analysis strategy.
In particular, in the case of a categorical outcome, we provide an analytic bound for the proxy-indicator relationship in the presence of zero inflation in a number of missing data models, and show that in some models our bound is sharp.
In more general cases, we show that the numeric approach for obtaining bounds detailed in \citet{duarte23automated} may be used instead.

Finally, we demonstrate an application of our method on simulated data, as well as a real world dataset on CLABSIs.

\section{Graphical Models of Missing Data and Zero Inflated Data}
\label{sec:prelim}

In this section we briefly review relevant existing works on
missing data, and describe difficulties posed by zero inflation.

\subsection{Missing data and identification}

Let $X^{(1)} = \{X_1^{(1)}, \ldots, X_n^{(1)}\}$ be a set of random variables (r.v.s) of interest.
Denote $\mathcal{X}^{(1)}_i$ as the state space of $X^{(1)}_i$, which we assume is categorical, and without loss of generality, includes the value $0$.
Samples of $X^{(1)}$ are systematically missing, with
true values being replaced by a special symbol ``?''.
To better represent missing data problems, it is convenient to use
two additional sets of r.v.s: the proxies $X = \{X_1, \ldots, X_n\}$, where each proxy $X_i \in X$ has the state space
$\mathcal{X}_i = \mathcal{X}_{i}^{(1)} \cup \{``?"\}$,
and the binary observability indicators $R = \{R_1, \ldots, R_n\}$.
Each proxy $X_i$ is deterministically defined in terms of the underlying variable $X^{(1)}_i$ and the observability indicator $R_i$ via the missing data version of the consistency rule: $X_i=X^{(1)}_i$ when $R_i=1$ and $X_i=``?"$ when $R_i=0$.
Thus, a variable $X_i^{(1)}$ may be described as "$X_i$ had it (hypothetically) been observed", i.e., a counterfactual. The superscript notation is deliberately chosen to make the connection to counterfactuals in causal inference explicit.
In addition to $X^{(1)}, R, X$, let $C$ represents other fully observed variables.

We define $R_{-i}$ as $\{ R_1, \ldots, R_{i-1}, R_{i+1}, \ldots, R_n \}$, $R_{<i}$ as $\{ R_1, \ldots, R_{i-1} \}$ and $R_{\geq i}$ as $\{ R_i \ldots, R_n \}$, with analogous subsets of $X$, $X^{(1)}$ 
defined similarly.
Following the nomenclature in \citet{bhattacharya20completeness, bhattacharya19mid}, we call $p(X^{(1)}, R, C)$ the full law, $p(R, X, C)$ the observed law, and $p(X^{(1)})$ the target law.
A missing data model is a set of distributions over the variables $\{X^{(1)}, R, X, C\}$ that satisfy the above consistency rule.

Following \citet{mohan13missing}, we consider missing data model defined using a class of directed acyclic graphs (DAGs) called missing data DAGs (m-DAGs).
Specifically, an m-DAG $\mathcal{G}(V)$ consists of nodes $V = \{X^{(1)}, R, X, C\}$.
Like all DAGs, m-DAGs only have directed edges and lack directed cycles, but also have a number of additional restrictions:
each proxy $X_i$ has exactly 2 incoming edges $X^{(1)}_i \rightarrow X_i \leftarrow R_i$ (due to consistency); there is no edge from any $X_i$ or $R_i$ to any $X^{(1)}_i$.
A joint $p(X^{(1)}, R, X)$ in the missing data model corresponding to the m-DAG $\mathcal{G}$ factorizes as
{\small
  \begin{align*}
    \prod_{V \in \{R, X^{(1)}\}} p(V \mid \operatorname{pa}_{\mathcal{G}}(V)) \prod_{X_i \in X} p(X_i \mid R_i, X^{(1)}_i)
  \end{align*}
}where all terms $p(X_i \mid R_i, X^{(1)}_i)$ are deterministic.
Using m-DAGs, one can represent many interesting missing data scenarios, see Fig.~\ref{fig:missing_data_example} for examples.

An important goal in missing data problems, prior to statistical inference, is to ensure the target parameter, which is generally some function of the target law, is identified from the observed law.
It follows by definition that the target law $p(X^{(1)})$ is identified if and only if the propensity score $p(R\mid X^{(1)})$ evaluated at $R=1$ is identified, while the full law $p(X^{(1)}, R)$ is identified if and only if the propensity score $p(R \mid X^{(1)})$ at all values of $R$ is identified.
While identification of the target law is still an open problem, \citet{bhattacharya20completeness}
showed a sound and complete method for identification of the full law $p(X^{(1)}, R)$ from the observed law $p(R, X)$ in missing data models represented by m-DAGs and hidden variable m-DAGs.

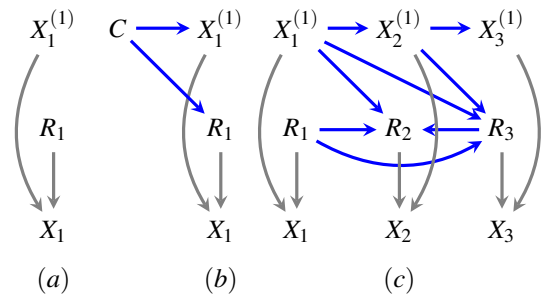
\begin{figure}[!htb]
  \centering
    \begin{tikzpicture}[scale=0.4,>=stealth, node distance=1.35cm]
      \tikzstyle{obs} = [circle, inner sep=1pt]
      \tikzstyle{unobs} = [circle, inner sep=0pt]
      \begin{scope}
        \path[->, very thick]
        node[unobs] (x11) {$X^{(1)}_1$}
        node[obs, below of=x11] (r1) {$R_1$}
        node[obs, below of=r1] (x1) {$X_1$}
  
        (x11) edge[gray, bend right] (x1)
        (r1) edge[gray] (x1)

        node[below of=x1, yshift=0.7cm] (l) {$(a)$}
        ;
      \end{scope}
      \begin{scope}[xshift=5.5cm]
        \path[->, very thick]
        node[unobs] (x11) {$X^{(1)}_1$}
        node[obs, below of=x11] (r1) {$R_1$}
        node[obs, below of=r1] (x1) {$X_1$}
        node[obs, left of=x11] (c) {$C$}
  
        (x11) edge[gray, bend right] (x1)
        (r1) edge[gray] (x1)
        (c) edge[blue] (x11)
        (c) edge[blue] (r1)

        node[below of=x1, yshift=0.7cm] (l) {$(b)$}

        ;
      \end{scope}      
      \begin{scope}[xshift=8.0cm]
        \path[->, very thick]
        node[unobs, ] (x11) {$X^{(1)}_1$}
        node[unobs, right of=x11] (x21) {$X^{(1)}_2$}
        node[unobs, right of=x21] (x31) {$X^{(1)}_3$}
        node[obs, below of=x11] (r1) {$R_1$}
        node[obs, below of=x21] (r2) {$R_2$}
        node[obs, below of=x31] (r3) {$R_3$}
        node[obs, below of=r1] (x1) {$X_1$}
        node[obs, below of=r2] (x2) {$X_2$}
        node[obs, below of=r3] (x3) {$X_3$}

        (x11) edge[blue] (x21)
        (x21) edge[blue] (x31)

        (r1) edge[blue] (r2)
        (r3) edge[blue] (r2)

        (x11) edge[blue] (r3)
        (r1) edge[blue, bend right] (r3)
  
        (x11) edge[gray, bend right] (x1)
        (r1) edge[gray] (x1)
        (x21) edge[gray, bend left] (x2)
        (r2) edge[gray] (x2)
        (x31) edge[gray, bend left] (x3)
        (r3) edge[gray] (x3)

        (x11) edge[blue] (r2)
        (x21) edge[blue] (r3)

        node[below of=x2, yshift=0.7cm] (l) {$(c)$}
        ;
      \end{scope}
    \end{tikzpicture}
  \caption{
    Missing data scenarios represented by m-DAG.
    Circle nodes denote observed variables, while others nodes are unobserved.
    Gray edges denote deterministic nature of $p(X_i \mid R_i, X^{(1)}_i)$ due to consistency.
    (a) $X^{(1)}_1$ is MCAR since $R_1 \Perp X^{(1)}$.
    (b) $X^{(1)}_1$ is MAR since $R_1 \Perp X^{(1)} \mid C$.
    (c) $X^{(1)}_1,X^{(1)}_2,X^{(1)}_3$ are MNAR, since observability indicators $R_1,R_2,R_3$ are are not independent of these variables,
    either marginally or given observed variables.
  }
  \label{fig:missing_data_example}
\end{figure}

\subsection{Zero Inflation Non-identifiability}

A zero inflated (ZI) model associated with an m-DAG is a variant of the missing data model associated with that m-DAG, with the following important difference: 
the missing data consistency relating variables $X^{(1)}_i\in X^{(1)}, X_i\in X, R_i \in R$ is replaced by a zero inflation version,
where $X_i = X^{(1)}_i$ if $R_i = 1$, and $X_i = 0$ if $R_i = 0$.
\footnote{Note that we consider ZI models with categorical state spaces only, unless stated otherwise.}

There are several important consequences of zero inflated consistency.
Firstly, both $X^{(1)}_i \in X^{(1)}$ and $X_i \in X$ take values in $\mathcal{X}_i$, and no variable in a ZI problems takes the value ``?''.
Secondly, as in missing data, the ZI-variable $X^{(1)}_i$ is counterfactual, and according to the ZI consistency rule, its true realizations are observed only when $R_i=1$.
In particular, if $X_i =x \neq 0$, we deduce $R_i=1$ and $X_i^{(1)} = x$.
However, since it is not possible to tell whether a realization $X_i = 0$ corresponds to the situation where $0$ is the true value of $X^{(1)}_i$, or corresponds to a censored realization of $X^{(1)}_i$,
$R_i$ is \emph{unobserved} whenever $X_i = 0$.
Moreover, while we still refer to $p(X^{(1)})$ and $p(X^{(1)}, R)$ as the target law and the full law, respectively, we will refer to
$p(R, X)$ as the zero-inflated law (ZI law), rather than the observed law, since $R$ is not always observed.
Thirdly, the ZI consistency imposes the following important restriction on the ZI law $p(R, X)$
\begin{itemize}[leftmargin=0cm]
  \item[] \textbf{(Z)} For every $i$ and $x \neq 0$, $p(R_{i}=0, X_{i}=x) = 0$.
\end{itemize}

We classify ZI models as ZI MCAR, ZI MAR, or ZI MNAR, if its missing data version is MCAR, MAR, or MNAR, respectively. Examples of ZI models are shown in Fig.~\ref{fig:zi_mcar_mar} and Fig.~\ref{fig:zi_mnar}.

Just as in missing data problems, the goal in ZI problems is to identify (a function of) the target law or the full law from the observed law and possibly additional objects. We focus on the full law identification in this paper.
Unsurprisingly, ZI problems are significantly harder than missing data problems, in the sense that both the target law
and the full law are non-parametrically non-identified even in the simplest setting (ZI MCAR), as shown by the following result.

\begin{lemma}[Non-identifiability]
  \label{lm:nonid}
    Given a ZI model associated with any m-DAG ${\cal G}$, both the target law $p(X^{(1)})$ and the full law $p(X^{(1)}, R, C)$ are non-parametrically non-identified.
\end{lemma}

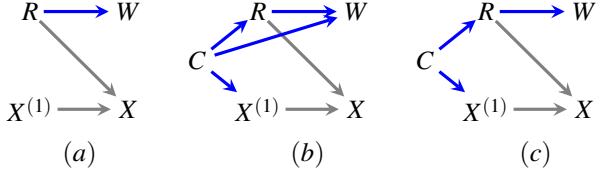
\begin{figure}[!htb]
  \centering
    \begin{tikzpicture}[>=stealth, node distance=1.3cm]
      \tikzstyle{obs} = [circle, inner sep=1pt]
      \tikzstyle{unobs} = [circle, inner sep=0pt]
      \begin{scope}[xshift=0cm]
        \path[->, very thick]
        node[unobs] (x1) {$X^{(1)}$}
        node[obs, right of=x1] (x) {$X$}
        node[unobs, above of=x1] (r) {$R$}
        node[obs, right of=r] (w) {$W$}

        (x1) edge[gray] (x)
        (r) edge[gray] (x)
        (r) edge[blue] (w)

        node[below of=x1, yshift=+0.7cm, xshift=+0.65cm] (l) {$(a)$}
        ;
      \end{scope}
      \begin{scope}[xshift=3cm]
        \path[->, very thick]
        node[unobs] (x1) {$X^{(1)}$}
        node[obs, right of=x1] (x) {$X$}
        node[unobs, above of=x1] (r) {$R$}
        node[obs, below of=r, xshift=-0.8cm, yshift=+0.65cm] (c) {$C$}
        node[obs, right of=r] (w) {$W$}

        (x1) edge[gray] (x)
        (r) edge[gray] (x)
        (r) edge[blue] (w)
        (c) edge[blue] (r)
        (c) edge[blue] (x1)
        (c) edge[blue] (w)

        node[below of=x1, yshift=+0.7cm, xshift=+0.65cm] (l) {$(b)$}
        ;
      \end{scope}
      \begin{scope}[xshift=6cm]
        \path[->, very thick]
        node[unobs] (x1) {$X^{(1)}$}
        node[obs, right of=x1] (x) {$X$}
        node[unobs, above of=x1] (r) {$R$}
        node[obs, below of=r, xshift=-0.8cm, yshift=+0.65cm] (c) {$C$}
        node[obs, right of=r] (w) {$W$}

        (x1) edge[gray] (x)
        (r) edge[gray] (x)
        (r) edge[blue] (w)
        (c) edge[blue] (r)
        (c) edge[blue] (x1)
        node[below of=x1, yshift=+0.7cm, xshift=+0.65cm] (l) {$(c)$}   
        ;
      \end{scope}      
    \end{tikzpicture}
    \caption{
    Examples of proxy-augmented ZI MCAR model (a) and ZI MAR models (b and c).
    \textbf{A1}, \textbf{A2} holds in (a),
    {\bf A1$^\dag$}, {\bf A2$^{\dag}$} hold in (b), and
    {\bf A1$^*$}, {\bf A2$^*$} hold in (c).
    Unlike missing data, indicator $R$ is partially observed.
    }
    \label{fig:zi_mcar_mar}
\end{figure}

\begin{figure}[!htb]
  \centering
    \begin{tikzpicture}[>=stealth, node distance=1.3cm]
      \tikzstyle{obs} = [circle, inner sep=0.5pt]
      \tikzstyle{unobs} = [circle, inner sep=0.5pt]
      \begin{scope}
        \path[->, very thick]
        node[unobs] (x11) {$X^{(1)}_1$}
        node[unobs, right of=x11] (x21) {$X^{(1)}_2$}
        node[unobs, below of=x11] (r1) {$R_1$}
        node[unobs, below of=x21] (r2) {$R_2$}
        node[obs,   above of=r1, xshift=-0.9546cm] (x1) {$X_1$}
        node[obs,   above of=r2, xshift=+0.9546cm] (x2) {$X_2$}
        node[obs,   below right of=r2, yshift=0.9546cm] (w2) {$W_2$}
        node[obs,   below left of=r1, yshift=0.9546cm] (w1) {$W_1$}

        (x11) edge[blue] (x21)

        (r1) edge[blue] (w1)
        (r2) edge[blue] (w2)
        
        (x11) edge[blue] (r2)
        (x21) edge[blue] (r1)
        
        (x11) edge[gray, bend right=0] (x1)
        (r1) edge[gray] (x1)
        (x21) edge[gray, bend left=0] (x2)
        (r2) edge[gray] (x2)

        node[below of=r1, xshift=+0.65cm, yshift=+0.7cm] (l) {$(a)$}
        ;
      \end{scope}
      \begin{scope}[xshift=4.2cm]
        \path[->, very thick]
        node[unobs] (x11) {$X^{(1)}_1$}
        node[unobs, right of=x11] (x21) {$X^{(1)}_2$}
        node[unobs, below of=x11] (r1) {$R_1$}
        node[unobs, below of=x21] (r2) {$R_2$}
        node[obs, above of=r1, xshift=-0.9546cm] (x1) {$X_1$}
        node[obs, above of=r2, xshift=+0.9546cm] (x2) {$X_2$}
        node[obs, below right of=r2, yshift=0.9546cm] (w2) {$W_2$}
        node[obs, below left of=r1, yshift=0.9546cm] (w1) {$W_1$}

        (x11) edge[blue] (x21)

        (r1) edge[blue] (r2)
        (r1) edge[blue] (w1)
        (r2) edge[blue] (w2)
  
        (x11) edge[blue] (r2)
        
        (x11) edge[gray, bend right=0] (x1)
        (r1) edge[gray] (x1)
        (x21) edge[gray, bend left=0] (x2)
        (r2) edge[gray] (x2)

        node[below of=r1, xshift=+0.65cm, yshift=+0.7cm] (l) {$(b)$}
        ;
      \end{scope}
    \end{tikzpicture}
  \caption{
    Examples of proxy-augmented ZI MNAR models.
      (a) ZI bivariate block-parallel model. (b) ZI bivariate block-sequencial MAR model.
  }
  \label{fig:zi_mnar}
\end{figure}
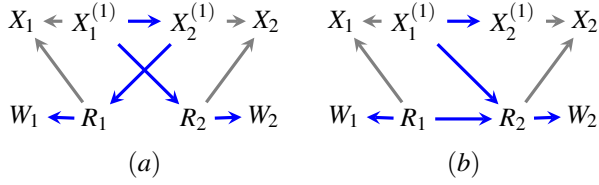

\section{Proxy-based Identification}

We first demonstrate our approach to proxy-based identification with the
simplest ZI missing data model, ZI MCAR, and generalize it to arbitrary
ZI m-DAG.

\subsection{Identification in the ZI MCAR Model}

Lemma~\ref{lm:nonid} implies that any identification method must rely on additional assumptions beyond those implied by the m-DAG.
To illustrate additional assumptions that will be employed, consider 
a simple ZI MCAR model with a single ZI variable $X^{(1)}$ taking values in $\{0, \ldots, K\}$, and the corresponding inflation indicator $R$, where $R \ci X^{(1)}$.

To simplify subsequent presentation, we will use the following notational shorthand: $p_{a_i \mid b_j}$ to mean $p(A=i \mid B=j)$, and $\mathbf{p}_{A \mid B}$ to mean the stochastic matrix whose elements are $p_{a_i \mid b_j}$.  Similarly for $p_{a_i, b_j}$ and $p_{A, B}$.
We also use a matrix multiplication shorthand, where $p_{a_i \mid b_j} p_{b_j \mid c_k} = p_{a_i \mid c_k}$ is taken to mean $\sum_{j} p(A=i | B=j) p(B=j | C=k) = p(A = i | C=k)$.

We will assume the existence of an observed binary proxy variable $W$ informative for $R$ with the following properties:
\begin{itemize}[leftmargin=0cm]
  \item[] \textbf{(A1)} $W \ci X^{(1)} \mid R$,
  \item[] \textbf{(A2)} The matrix $\mathbf{p}_{W \mid R}$ is invertible.
\end{itemize}
Note that since $W$ and $R$ are binary, \textbf{A2} is equivalent to $p_{w_0 \mid r_0} \neq p_{w_0 \mid r_1}$.
Due to the existence of the proxy variable $W$, we call this ZI MCAR model "proxy-augmented", whose graph is shown in Fig.~\ref{fig:zi_mcar_mar} (a).

\citet{kuroki14measurement} considered assumptions {\bf A1} and {\bf A2} in the context of obtaining identification of causal effects in the presence of unobserved confounding.
In that work, the proxy variable $W$ was related to an unobserved categorical variable which was a common cause of the treatment and outcome variables.

In this paper, we adopt the method of \citet{kuroki14measurement}
to express the ZI law $p(R, X)$ in terms of the observed law $p(X, W)$ and the conditional distribution $p(W \mid R)$.
In addition to \textbf{A1} and \textbf{A2}, the Kuroki-Pearl method requires that the observed law $p(X, W)$ and $p(W \mid R)$ are from the same full law (e.g. compatible), and $p(W \mid R)$ is known.

To see that point identification is then possible, we write $p(wx) = \sum_r p(w \mid r) p(r,x)$ in matrix form
{\small
\begin{equation}
\setlength\arraycolsep{2pt}
\begin{aligned}
\underbrace{\begin{pmatrix} p_{w_0, x_0} & \cdots & p_{w_0, x_K} \\ p_{w_1, x_0} & \cdots & p_{w_1, x_K} \end{pmatrix}}_{\mathbf{p}_{WX}}
&= \underbrace{\begin{pmatrix} p_{w_0 \mid r_0} & p_{w_0 \mid r_1} \\ p_{w_1 \mid r_0} & p_{w_1 \mid r_1} \end{pmatrix}}_{\mathbf{p}_{W \mid R}}
\:
\underbrace{\begin{pmatrix} p_{r_0,  x_0} & \cdots & 0\\ p_{r_1, x_0} & \cdots & p_{r_1,  x_K} \end{pmatrix}}_{\mathbf{p}_{RX}},
\end{aligned}
\end{equation}
}where the $0$ entry in $\mathbf{p}_{RX}$ is due to the restriction {\bf Z}.
Since $\mathbf{p}_{W \mid R}$ is invertible, we can solve for $\mathbf{p}_{RX}$ by $[\mathbf{p}_{W \mid R}]^{-1} \mathbf{p}_{WX}$,
leading to the following result.
\begin{theorem}[ZI law restoration in ZI MCAR]
  \label{thm:rest-zcar}
  For the ZI MCAR model in Fig.~\ref{fig:zi_mcar_mar} (a) under {\bf A1}, {\bf A2},
  the ZI law $p(R, X,W)$ is point identified given the observed law $p(X, W)$ and
  a compatible proxy-indicator conditional distribution $p(W \mid R)$,
  as follows
  \begin{equation}
    \label{eq:kp_id}
    p(r,x,w) = p(w \mid r) \left[\mathbf{p}_{W \mid R}^{-1} \mathbf{p}_{WX}\right]_{r,x}.
  \end{equation}
\end{theorem}
After the ZI law $p(R, X,W)$ is identified, the full law is identified,
$p(X^{(1)}, R, W) = p(X, W \mid R=1) p(R)$, by standard assumptions of the MCAR model.

\begin{remark}
There are two difficulties with this result.  First, since $R$ is potentially unobserved, it is not always reasonable to specify the true $p(W \mid R)$ in applications.  Second, since our working model corresponds to a hidden variable DAG, the model imposes restrictions on
the pair $\left(p(X, W), p(W \mid R)\right)$, meaning that not every potential distribution $p(W \mid R)$ would be consistent with the observed data law under our model.
Using inconsistent $p(W \mid R)$ in the matrix inversion equation places us outside the model, and can yield inconsistent results, such as invalid negative probabilities $p(R, X)$.
Examples of such an inconsistency is provided in the Appendix.
\citet{kuroki14measurement} noted the latter issue in the context of causal inference, but did not provide bounds.
\end{remark}

\subsection{Proxy-based identification in general ZI missing data models}

In this section, we generalize our previous proxy-based approach to an
arbitrary graphical ZI model corresponding to an m-DAG, given that the full
law is point identified in the missing data model associated to that m-DAG.

Consider any ZI model associated with an arbitrary m-DAG,
with a set of fully observed covariates $C$, a set of ZI variables $X^{(1)} = \{ X^{(1)}_1, \ldots, X^{(1)}_n \}$,
inflation indicators $R = \{ R_1, \ldots, R_n \}$, observed versions $X = \{ X_1, \ldots, X_n \}$ for variables in $X^{(1)}$, and proxies $W = \{ W_1, \ldots, W_n \}$ for variables in $R$.

We make the following assumptions which generalize {\bf A1} and {\bf A2}:

\begin{itemize}[leftmargin=0cm]
  \item[] \textbf{(A1$^*$)} $\forall i$, $W_i \ci X^{(1)},C,{R_{-i}} \mid R_i$.
  \item[] \textbf{(A2$^*$)} The matrix $\mathbf{p}_{W \mid R}$ is invertible.
\end{itemize}
In addition, we will provide alternatives to {\bf A1$^*$} and {\bf A2$^*$} which allow the proxies $W$ to potentially depend on $C$:
\begin{itemize}[leftmargin=0cm]
    \item[] \textbf{(A1$^{\dag}$)} $\forall i$, $W_i \ci X^{(1)},R_{-i} \mid C,R_i$.
    \item[] \textbf{(A2$^{\dag}$)} The matrix $\mathbf{p}_{W \mid R,c}$ is invertible for every value $c$.
\end{itemize}

The identification strategy we adopt proceeds in two stages:
\begin{enumerate}
  \item \textbf{ZI law restoration}:
  point identify (if true $p(W \mid R)$ is known) or partially identify the ZI law $p(R, X,W,C)$ from the observed law $p(X, W, C)$. 
  \item \textbf{Downstream identification}: identify the full law $p(X^{(1)}, R, W, C)$ from the ZI law $p(R, X,W,C)$.
\end{enumerate}
Since every $R_i$ is unobserved whenever $X_i = 0$ in ZI problems, the purpose of the first stage is to recover the ZI law involving $R$ and other observed variables.
Under mentioned proxy assumptions and knowledge of the true $p(W \mid R)$, point identification of this law is possible. Otherwise, partial identification bounds are computed.
If point or partial identification is possible,
variables in $R$ may now be treated as observed data, and the problem is reduced to classical identification in missing data model.  In particular, we adopt the sound and complete identification procedure described by \citet{bhattacharya20completeness} to point identify the full law in the second stage.

While we focus on non-parametric point identification results for the full law, one could instead employ any point or partial identification procedure developed for missing data problems for the second stage.
We leave these types of extensions to future work.

\subsubsection{ZI law restoration}
\label{subsec:zi_law_rest}

Under the proxy assumptions,
we have the following identification result, which generalizes Theorem~\ref{thm:rest-zcar}.

\begin{theorem}[ZI law restoration]
  \label{thm:rest-general}
  Given a ZI model satisfying assumptions {\bf A1$^*$} and {\bf A2$^*$} (or {\bf A1$^{\dag}$} and {\bf A2$^{\dag}$}),
  the ZI law $p(R, X,W,C)$ is point identified given the observed law $p(X, W, C)$ and a compatible proxy-indicator conditional distribution $p(W \mid R)$ (OR $p(W \mid R,C)$),
  \begin{equation}
    \label{eq:kp_id_general}
  \begin{aligned}
    & p(r,x,w,c) = p(w \mid r)\!\! \left[\mathbf{p}_{W \mid R}^{-1} \mathbf{p}_{WXC}\right]_{r,x,c}
    \!\!\!\!\!
    \!\!\!\!\!
    \text{ under {\bf A1$^*$},{\bf A2$^*$}} \\
    & p(r,x,w,c) = p(w \mid r,c)\!\! \left[\mathbf{p}_{W \mid R,C}^{-1} \mathbf{p}_{WXC}\right]_{r,x,c}
    \!\!\!\!\!
    \!\!\!\!\!
    \text{ under {\bf A1$^{\dag}$},{\bf A2$^{\dag}$}}.
  \end{aligned}
  \end{equation}
\end{theorem}

Theorem~\ref{thm:rest-general} suffers from the same issue as Theorem~\ref{thm:rest-zcar}: it is unlikely that the true distribution $p(W \mid R)$ (or $p(W \mid R,C)$) will always be available, and given a candidate distribution, it is not obvious to verify that it is compatible with the model and the observed law.

If the true $p(W \mid R)$ (or $p(W \mid R,C)$) is not given, we must find the set of compatible $p(W \mid R)$ (or $p(W \mid R,C)$) distributions to the model and the observed law.
In general, bounds on $p(W \mid R)$ (or $p(W \mid R,C)$) may be computed numerically by encoding the model as a system of polynomial equations and finding extrema of this system using polynomial programming. A method for solving such systems of equations using a primal/dual method is described in \citet{duarte23automated}.
These bounds lead to a natural sensitivity analysis strategy according to our two stage approach. Particularly, each compatible $p(W \mid R)$ in the bounds implies a valid ZI law by Theorem~\ref{thm:rest-general},
which in turn implies a full law by Proposition~\ref{prop:full-law-id} of the next section.
In Section~\ref{sec:experiments}, we conduct a grid search of the compatible set to illustrate this point.

While numeric bound computation is a general approach, finding such bounds is computationally challenging due to the need to solve polynomial programs.  Fortunately, we show that
in certain ZI models, it is possible to derive analytic bounds on $p(W \mid R)$ (or $p(W \mid R,C)$), instead. 
We also show that these bounds are sharp in some cases.

\subsubsection{Downstream identification}

After the ZI law $p(R, X,W,C)$ is recovered in the restoration step,
one may consider this law as the "observed law" in the missing data problem corresponding to the same m-DAG, and invoke missing data identification to obtain the full law $p(X^{(1)},R,W,C)$.
We note that this second identification stage is not precisely the same as that for standard missing data problems, because identification relies on consistency, and consistency under ZI differs from missing data consistency whenever $R=0$.

Fortunately, consistency when $R=1$ coincides in ZI problems and missing data problems, and, as the following result shows, suffices for identification.

\begin{proposition}[ZI full law identification]
The full law $p\left(X^{(1)}, R, W, C\right)$ exhibiting zero inflation that is Markov relative to an m-DAG $\mathcal{G}$ is identified given the ZI law $p(R, X, W, C)$ if and only if $\mathcal{G}$ does not contain edges of the form $X_i^{(1)} \rightarrow R_i$ (no self-censoring) and structures of the form $X_j^{(1)} \rightarrow R_i \leftarrow R_j$ (no colluders), and the positivity assumption holds.
Moreover, the identifying functional for the full law coincides with the functional given in \citet{malinsky2021semiparametric}.
\label{prop:full-law-id}
\end{proposition}

\subsection{Partial Identification in ZI MCAR}
\label{subsec:pid_zi_mcar}

In this subsection, we relax the requirement that the true $p(W \mid R)$ must be given in the ZI law restoration step, and provide bounds for this conditional distribution in the proxy-augmented ZI MCAR model.

Consider the proxy-augmented ZI MCAR model in Fig.~\ref{fig:zi_mcar_mar} (a).
This model is equivalently described by the following model $\mathcal{P}$,
satisfying {\bf Z}, \textbf{A1}, and \textbf{A2}:
\begin{equation}
  \label{eq:model}
  \begin{aligned}
    \mathcal{P} = \left\{
      \begin{aligned}
      (\mathbf{q}_{W \mid R}, \mathbf{q}_{RX}):
      &&& \!\textstyle \mathbf{q}_{W \mid R} \geq 0, \sum_w q_{w \mid r} = 1, \forall r, \\
      &&& \!\textstyle \mathbf{q}_{RX} \geq 0, \sum_{rx} q_{rx} = 1, \\
      &&& \!\textstyle \forall x \neq 0 (q_{r_0 x} = 0); \: q_{w_0 \mid r_0} \neq q_{w_0 \mid r_1}. \\
      \end{aligned}
    \right\}
  \end{aligned}
\end{equation}
Given an observed law $p(X, W)$, we are interested in the following subset
$\mathcal{Q} \subseteq \mathcal{P}$ of distributions yielding the
observed law,
\begin{equation}
  \label{eq:compatible_set}
    \mathcal{Q} =
    \left\{
        (\mathbf{q}_{W \mid R}, \mathbf{q}_{RX}) \in \mathcal{P} : \:
        \mathbf{q}_{W \mid R} \mathbf{q}_{RX} = \mathbf{p}_{WX}
    \right\}.
\end{equation}

In particular, our goal is finding all $\mathbf{q}_{W \mid R} \in \mathcal{Q}$, which is the \textbf{partial identification} of $q(W \mid R)$ w.r.t. the given observed law.
This is equivalent to projecting $\mathcal{Q}$ onto the probability simplex of $\mathbf{q}_{W \mid R}$.
From (\ref{eq:compatible_set}) and (\ref{eq:model}), one way to check whether an invertible $\mathbf{q}_{W \mid R} \in \mathcal{Q}$
is to compute $\mathbf{q}_{RX} = (\mathbf{q}_{W \mid R})^{-1} \mathbf{p}_{WX}$ and
check $(\mathbf{q}_{W \mid R}, \mathbf{q}_{RX}) \in \mathcal{Q}$.
First, $\mathbf{q}_{RX}$ must be a stochastic matrix for any problem under \textbf{A1} and \textbf{A2}. Second, $\mathbf{q}_{RX}$ must also satisfy ZI-consistency constraint \textbf{Z}.
If these conditions are true, there is a joint distribution in the model generates both $\mathbf{q}_{W \mid R}$ and $\mathbf{p}_{WX}$, and they are said to be \textbf{compatible}.
After the compatible set of $\mathbf{q}_{W \mid R}$ is derived, the partial identification of $\mathbf{q}_{RX}$ could be obtained using
(\ref{eq:kp_id}).

We note that \textbf{Z}
implies, for all $x \neq 0$, $q_{x} = q_{r_1, x}$, so $q_{r_1 \mid x} = 1$.
Then by considering $q_{w_0 \mid r_0} q_{r_0, x} + q_{w_0 \mid r_1} q_{r_1, x} = p_{w_0, x}$,
we obtain point identification $q_{w_0 \mid r_1} = p_{w_0 \mid x_1}$
and the marginal constraints $\forall x \neq 0, p_{w_0 \mid x} = p_{w_0 \mid x_1}$.  Note that these constraints may be used to design a falsification test of the model.

However, $q_{w_0 \mid r_0}$ is not identified, and its bounds must be obtained by solving the following polynomial program:
\begin{equation}
  \label{eq:optim_mcar}
  \begin{aligned}
    \max_{ q_{w_0 \mid r_0}} &&& \pm q_{w_0 \mid r_0}
    \\
    \text{s.t.}
    &&& \mathbf{q}_{W \mid R} \mathbf{q}_{RX} = \mathbf{p}_{WX}, \\
    &&& \textstyle \mathbf{q}_{W \mid R} \geq 0, \: \forall r (\sum_{w} q_{w|r} = 1), \: q_{w_0 \mid r_0} \neq q_{w_0 \mid r_1}, \\
    &&& \textstyle \mathbf{q}_{RX} \geq 0, \: \sum_{rx} q_{rx} = 1, \: q_{w_0 \mid r_1} = p_{w_0 \mid x_1}. \\
  \end{aligned}
\end{equation}

Since both $q_{w_0 | r_0}$ and $ \mathbf{q}_{RX}$ are unknowns,
the above system of equations corresponds to a quadratic program, which is difficult to solve in general.

However, it is possible to transform this optimization into an equivalent linear program with the following observations:
\begin{enumerate}
  \item A specific solution to $\mathbf{q}_{RX}$ is not required. One merely needs to check if $\mathbf{q}_{W \mid R}^{-1} \mathbf{p}_{WX}$ is a stochastic matrix.
  \item If $\mathbf{q}_{W \mid R} \mathbf{q}_{RX} = \mathbf{p}_{WX}$, where all matrices are non-negative, $\mathbf{p}_{WX}$ sum to 1 and $\mathbf{q}_{W \mid R}$ sum to 1, then $\mathbf{q}_{RX}$ sum to 1. The proof of this fact is in the Appendix.
  \item  The inverse $\left[\mathbf{q}_{W \mid R}\right]^{-1}$ is
    \begin{equation}
        \textstyle
        \frac{1}{q_{w_0 \mid r_0} - q_{w_0 \mid r_1}} \begin{pmatrix} 1-q_{w_0 \mid r_1} & - q_{w_0 \mid r_1} \\ q_{w_0 \mid r_0}-1 & q_{w_0 \mid r_0} \end{pmatrix}
        \label{eqn:inverse}
    \end{equation}
\end{enumerate}

Observations 1 and 2 imply that checking compatibility
involves only checking non-negativity of $\mathbf{q}_{W \mid R}^{-1} \mathbf{p}_{WX}$, reducing the unknowns in our optimization problem to only $q_{w_0 \mid r_0}$.
Checking $\mathbf{q}_{W \mid R}^{-1} \mathbf{p}_{WX}$ is still non-linear in $\mathbf{q}_{W \mid R}$,
but (\ref{eqn:inverse}) suggests an equivalent procedure consisting of two separate problems where
$q_{w_0 \mid r_0} > q_{w_0 \mid r_1}$ or $q_{w_0 \mid r_0} < q_{w_0 \mid r_1}$, respectively.
Concretely, for each case $s=1$ and $s=-1$, we consider 2 linear programs
\begin{equation}
  \begin{aligned}
    \max_{ q_{w_0 \mid r_0} } &&& \pm q_{w_0 \mid r_0} \\
    \text{s.t.}
      &&& s \cdot \begin{pmatrix} 1-q_{w_0 \mid r_1} & - q_{w_0 \mid r_1} \\ q_{w_0 \mid r_0}-1 & q_{w_0 \mid r_0} \end{pmatrix}  \mathbf{p}_{WX}  \geq \mathbf{0}, \\
      &&& s \cdot q_{w_0 \mid r_0} > s \cdot q_{w_0 \mid r_1}, 0 \leq q_{w_0 \mid r_0} \leq 1, \\
      &&& q_{w_0 \mid r_1} = p_{w_0 \mid x_1}. \\
  \end{aligned}
\end{equation}

These problems could be solved analytically using fast linear program solvers, yielding the following partial identification result for $p(W \mid R)$. A detailed proof is in the Appendix.

\begin{theorem}[ZI MCAR compatibility bound]
  \label{thm:zi_mcar_bound}
  Consider a ZI MCAR model in Fig.~\ref{fig:zi_mcar_mar} (a) under proxy assumptions {\bf A1}, {\bf A2},
  with categorical $X$ and binary $R, W$.
  Given a consistent observed law $p(X, W)$ satisfying positivity assumption, $\forall x, p(x) > 0$,
  the set of compatible proxy-indicator conditionals $q(W \mid R)$ is given by
  \begin{align*}
    q_{w_0 \mid r_1} &= p_{w_0 \mid x_1}\\
    q_{w_0 \mid r_0} &\in
    \begin{cases}
    [ p_{w_0 \mid x_0}, 1] \text{ if }p_{w_0 \mid x_0} > p_{w_0 \mid x_1}\\
    [ 0, p_{w_0 \mid x_0}] \text{ if }p_{w_0 \mid x_0} < p_{w_0 \mid x_1}\\
    (0, 1) \setminus \{ p_{w_0 \mid x_0} \}\text{ if } p_{w_0 \mid x_0} = p_{w_0 \mid x_1}
    \end{cases}
  \end{align*}
  These bounds are sharp.
  Moreover, if $p_{w_0 \mid x_0} = p_{w_0 \mid x_1}$, $p(X, W)$ must satisfy $0 < p_{w_0 \mid x_0} < 1$,
  and zero inflation does not occur, i.e., $q(R=0)=0$.
\end{theorem}

\subsection{Partial Identification in ZI MAR}
\label{subsec:pid_zi_mar}

We compute analytical bounds for two versions of the proxy-augmented ZI MAR model, illustrated  in Fig.~\ref{fig:zi_mcar_mar} (b) and (c).
The first model has $C \rightarrow W$ and satisfies {\bf A1$^{\dag}$} and {\bf A2$^{\dag}$},
while $C \not\rightarrow W$ in the second model, and the proxy assumptions are {\bf A1$^{*}$} and {\bf A2$^{*}$}.

In the first proxy-augmented ZI MAR model,
the set of compatible $p_{W \mid R, C}$ is given by the Cartesian product of the independently determined ZI MCAR bounds for each value $c$.
This leads to the following direct analogue of Theorem~\ref{thm:zi_mcar_bound}.
The proof is deferred to the Appendix.

\begin{theorem}[ZI MAR compatibility bound 1]
  \label{thm:zi_mar_bound_1}
  Consider a ZI MAR model in Fig.~\ref{fig:zi_mcar_mar} (b)
  under proxy assumptions {\bf A1$^{\dag}$} and {\bf A2$^{\dag}$}, 
  with categorical $X, C$ and binary $R, W$.
  Given a consistent observed law $p(X, W, C)$ satisfying positivity assumption, $\forall x, c, p(x, c) > 0$,
  the set of compatible proxy-indicator conditional distributions $q(W \mid R,C)$ is given by, for each  value $c$,
  \begin{align*}
      q_{w_0 \mid r_1, c} &= p_{w_0 \mid x_1, c}\\
      q_{w_0 \mid r_0, c} &\in
      \begin{cases}
      [ p_{w_0 \mid x_0, c}, 1] \text{ if }p_{w_0 \mid x_0, c} > p_{w_0 \mid x_1, c}\\
      [ 0, p_{w_0 \mid x_0, c}] \text{ if }p_{w_0 \mid x_0, c} < p_{w_0 \mid x_1, c}\\
      (0, 1) \setminus \{ p_{w_0, \mid x_0, c} \}\text{ if } p_{w_0 \mid x_0, c} = p_{w_0 \mid x_1, c}
      \end{cases}
  \end{align*}
  These bounds are sharp.
  Moreover, if $p_{w_0 \mid x_0, c} = p_{w_0 \mid x_1, c}$, $p(X, W, C)$ must satisfy $0 < p_{w_0 \mid x_0, c} < 1$,
  and zero inflation does not occur for stratum $C=c$, i.e., $q(R=0 \mid c)=0$.
\end{theorem}

On the other hand, the compatibility bound for the second ZI MAR model is the intersection of the ZI MCAR bounds for each values $c$. The proof is deferred to the Appendix.
\begin{theorem}[ZI MAR compatibility bound 2]
  \label{thm:zi_mar_bound_2}
  Consider a ZI MAR model in Fig.~\ref{fig:zi_mcar_mar} (c)
  under proxy assumptions {\bf A1$^*$} and {\bf A2$^*$},
  with categorical $X, C$ and binary $R, W$.
  Given a consistent observed law $p(X, W, C)$ satisfying positivity assumption, $\forall x, c, p(x, c) > 0$,
  the set of compatible proxy-indicator conditional distributions $q(W \mid R)$ is given by
  \begin{align*}
      q_{w_0 \mid r_1} &= p_{w_0 \mid x_1}\\
      q_{w_0 \mid r_0} &\in
      \begin{cases}
      [ \max_{c} p_{w_0 \mid x_0, c}, 1] \text{ if } \exists \tilde{c}, p_{w_0 \mid x_0, \tilde{c}} > p_{w_0 \mid x_1},\\
      [ 0, \min_c p_{w_0 \mid x_0, c}] \text{ if } \exists \tilde{c}, p_{w_0 \mid x_0, \tilde{c}} < p_{w_0 \mid x_1},\\
      (0, 1) \setminus \{ p_{w_0 \mid x_1} \}\text{ if } \forall c, p_{w_0 \mid x_0, c} = p_{w_0 \mid x_1}.
      \end{cases}
  \end{align*}
  These bounds are sharp.
  Moreover, if $\forall c, p_{w_0 \mid x_0, c} = p_{w_0 \mid x_1}$, $p(X, W, C)$ must satisfy $\forall c, 0 < p_{w_0 \mid x_0, c} < 1$,
  and zero inflation does not occur, i.e., $q(R=0)=0$.
\end{theorem}
Note that the first two cases are mutually exclusive due to the following lemma.
\begin{lemma}
    \label{lm:zi_mar_obs_constraint}
    For a ZI MAR model in Fig.~\ref{fig:zi_mcar_mar} (b) under 
    {\bf A1$^\dag$} and {\bf A2$^\dag$}, the observed law $p(X, W, C)$ obeys
    \begin{equation}
        \label{eq:zi_mar_obs_constraint_1}
        \forall c, \forall x \neq 0, p_{w_0 \mid x, c} = p_{w_0 \mid x_1, c}.
    \end{equation}
    For a ZI MAR model in Fig.~\ref{fig:zi_mcar_mar} (c) under 
    {\bf A1$^*$} and {\bf A2$^*$}, the observed law $p(X, W, C)$ obeys
    \begin{align}
        \label{eq:zi_mar_obs_constraint_2}
        & \forall c, \forall x \neq 0, p_{w_0 \mid x, c} = p_{w_0 \mid x_1}, \\
        \notag
        & \text{either } \forall c \left(p_{w_0 \mid x_0, c} \leq p_{w_0 \mid x_1} \right) \text{ or } \forall c \left(p_{w_0 \mid x_0, c} \geq p_{w_0 \mid x_1}\right).
    \end{align}
\end{lemma}
Note that, as before, the marginal constraints described may be used to design a model falsification test.

\subsection{Partial Identification In ZI MNAR}
\label{subsec:pid_zi_mnar}

Consider the ZI version of any MNAR model represented by an m-DAG where the target law is identified.
In missing data, an important subclass of such models are submodels of the no-self-censoring model in
\citet{malinsky2021semiparametric} due to the results in \citet{bhattacharya20completeness}.
The ZI versions of such models exhibit a crucial complication not found in previously discussed ZI models, namely that multiple variables may be zero inflated. 
For these models, we posit a set of proxies $W = \{ W_1, \ldots, W_n \}$ corresponding to $R = \{ R_1, \ldots, R_n \}$,
and assume assumptions {\bf A1$^*$}, {\bf A2$^*$} in Section~\ref{subsec:zi_law_rest} are satisfied.
Fig.~\ref{fig:zi_mnar} (a) and (b) show two bivariate examples of such models.
We use the short hand $p_{w_{ka} \mid x_{kb}} = p(W_k=a \mid X_{k}=b)$.

Given observed law $p(X, W, C)$, we seek the compatible set of $\{q(W_k \mid R_k)\}_{k=1}^n$,
whose elements allow restoration of $p(R, X,W,C)$ via Theorem~\ref{thm:rest-general}.
Although sharp bounds for $\{q(W_k \mid R_k)\}_{k=1}^n$ are unknown, the ZI MAR partial identification procedure could be applied to each $R_k$ independently to obtain bounds for $q_{w_{k0} \mid r_{k0}}$.  Moreover, due to the usual properties of ZI, $q_{w_{k0} \mid r_{k1}}$ is point identified for each $k$.

For each $k$, we apply Theorem~\ref{thm:zi_mar_bound_2} with $X_k, R_k, W_k$ being $X, R, W$, respectively, and $Z_k \triangleq \{X, W, C\} \setminus \{X_k, W_k\}$ being the covariates $C$.
These bounds are not sharp as structural constraints of the MNAR model are not considered.  However, these bounds are valid in the sense that the Cartesian product of these bounds contains the true model compatible set of distributions $\{ p(W_k \mid R_k) : k = 1, \ldots n \}$.

In addition, we note that (\ref{eq:general_obs_constraint}) below hold in the observed law under our model, and may be used as falsification test for our ZI model.
\begin{lemma}
    \label{lm:general_obs_constraint}
    Consider any ZI model in Section~\ref{subsec:zi_law_rest} under {\bf A1$^*$} and {\bf A2$^*$}. Denote $Z_k \triangleq \{X, W, C\} \setminus \{W_k, X_k\}$. The observed law $p(X, W, C)$ must satisfy, for each $k$,
    \begin{align}
        \label{eq:general_obs_constraint}
        & \forall z_k, \forall x \neq 0 p_{w_{k0} \mid x_{k}=x, z_k} = p_{w_{k0} \mid z_{k1}}, \\
        \notag
        & \forall z_k \left( p_{w_{k0} \mid x_{k0}, z_k} \leq p_{w_{k0} \mid x_{k1}} \right) \text{ or } \forall z_k \left(p_{w_{k0} \mid x_{k0}, z_k} \geq p_{w_{k0} \mid z_{k1}}\right).
    \end{align}
\end{lemma}

\subsection{Identification Given A Known Zero Inflation Probability}

For ZI MCAR models in Theorem~\ref{thm:zi_mcar_bound} and ZI MAR model in Theorem~\ref{thm:zi_mar_bound_2}, we provided the identification $q_{w_0 \mid r_1} = p_{w_0 \mid x_1}$ and the bounds for $q_{w_0 \mid r_0}$, which lead to partial identification of the full law $p(X^{(1)}, R, X, W, C)$.

If $q_{w_0 \mid r_0}$ is known a priori, the full law is point identified.  Alternatively, point identification of the full law may be obtained if the zero inflation probability, or $p(R=0)$, is known.

This is because the joint distribution $p(W,R)$ for binary $W,R$ has dimension $3$, and one (variationally dependent) parameterization for this joint is via the following $3$ parameters $p(R = 0)$, $p(W = 0)$, $q_{w_0 \mid r_1}$. This is easy to see by noting that we can compute $p(R = 0, W = 0) = q_{w_0 \mid r_1} (1 - p(R = 0))$, and $p(R = 0)$, $p(W = 0)$, and $p(R = 0, W=0)$ are the M\"obius parameters for $p(W,R)$ \citep{evans14markovian}.

In particular, we have the following: $q_{w_0 \mid r_0} = \frac{p_{w_0} - q_{w_0 \mid r_1} p_{r_1}}{p_{r_0}} = \frac{p_{w_0} - p_{w_0 \mid x_1} p_{r_1}}{p_{r_0}}$, which in turns implies point identification of the full law.

Note that not every zero inflation probability $p(R = 0)$ is compatible with the model.  This is easily seen by noting that the M\"obius parameterization is variationally dependent, and two parameters, namely $p(W = 0)$ and $q_{w_0 \mid r_1}$, are known.  Howevre, our derived bounds for $q_{w_0 \mid r_1}$ naturally imply bounds for $p(R = 0)$, with sharp bounds for the former implying sharp bounds for the latter.

\section{Experiments}
\label{sec:experiments}

We confirmed the validity of our analytical results for inflated zero models by sampling data generating processes (DGPs), and numerical methods.
In addition, we used our methods to perform sensitivity analyses on CLABSI data. Details of these experiments are in the Appendix.
The code could be found at \url{https://github.com/trungpq-ci/zero-inflation-bounds}.

\subsection{Bound Validity In Random DGPs}
\label{subsec:sims}

We verify the results of Theorem~\ref{thm:zi_mcar_bound}, Theorem~\ref{thm:zi_mar_bound_2} and related observed law constraints
by randomly generating DGPs in models we described.
We generated $10^8$ DGPs in the model in Fig.~\ref{fig:zi_mcar_mar} (a), satisfying ZI-consistency, {\bf A1}, {\bf A2}, and $10^8$ DGPs in the model in Fig.~\ref{fig:zi_mcar_mar} (b), satisfying ZI-consistency, {\bf A1$^*$}, {\bf A2$^*$}.
For both cases, we verified identification of $q_{w_0 \mid r_1}$ and the bounds for $q_{w_0 \mid r_0}$ as predicted by the corresponding theorem.
For the bounds, two tests were conducted
\begin{enumerate}
    \item \textit{Bound validity}: is the true $p_{w_0|r_0}$ inside the bounds?
    \item \textit{Model consistency}: grid search the bound, compute $p(r,x)$ (or $p(r,x,c)$) according to Theorem~\ref{thm:rest-zcar} (or ~\ref{thm:rest-general}), and verify that these are probability distributions.
\end{enumerate}
Additionally, for ZI MAR, we checked marginal constraints in (\ref{eq:zi_mar_obs_constraint_2}).
We found that all considered results held up to floating point precision in every single DGP.

\subsection{Bounds By Numerical Methods}

We compared our analytical bounds with numerical bounds computed using \texttt{autobounds} package in \citet{duarte23automated} for a subset of DGPs used for verification of bound validity in Section~\ref{subsec:sims}.

In particular, 20 DGPs were randomly selected for each model (ZI MCAR and ZI MAR), and their observed laws were computed.
For each DGP, 2 polynomial programs were constructed,
whose objective functions are maximizing or minimizing $p(W=0 | R=0)$, respectively, and whose constraints are (i) structural constraints from the corresponding graph, (ii) probability constraints, (iii) ZI-consistency constraint, (iv) constraints resulted from the structure imposed on the observed law by the structure of the full law.
The solutions to these programs are the numerical lower and upper bounds of $p(W=0 | R=0)$.
We refer reader to \citet{duarte23automated} for details of the program's construction and the methods used by the polynomial program solver.

For all DGPs, the numerical bounds coincided with our analytical bounds
up to the $4$th decimal place.
Since the algorithm in the \texttt{autobounds} package is an anytime algorithm, our analytic bounds were always contained inside the numerical bounds.
Table \ref{tab:num_bounds} shows a selection of these results.

\begin{table}
\centering
\footnotesize
    \begin{tabular}{|c| c | c | c | c | c|} 
        \hline
        MCAR  &  lb      &  ub      &  num lb  &  num ub  &  $p_{w_0|r_0}$ \\
        \hline\hline
        0     &  0.5564  &  1       &  0.5564  &  1       &  0.8207  \\
        \hline
        1     &  0.3578  &  1       &  0.3578  &  1       &  0.4936  \\
        \hline
        2     &  0       &  0.5206  &  0       &  0.5206  &  0.4536  \\
        \hline
        3     &  0.6064  &  1       &  0.6064  &  1       &  0.6826  \\
        \hline \hline
        
        MAR   &  lb      &  ub      &  num lb  &  num ub  &  $p_{w_0|r_0}$ \\
        \hline\hline
        0     &  0       &  0.4290  &  0       &  0.4290  &  0.4132  \\
        \hline
        1     &  0.8346  &  1       &  0.8346  &  1       &  0.8486  \\
        \hline
        2     &  0       &  0.3404  &  0       &  0.3404  &  0.3192  \\
        \hline
        3     &  0.3002  &  1       &  0.3002  &  1       &  0.5155  \\
        \hline
    \end{tabular}
  \caption{
    Comparison between our analytical lower and upper bound (\textit{lb}/\textit{ub})
    to numerical bounds (\textit{num lb}/\textit{num ub}) for a randomly selected set of DGPs. True $p_{w_0|r_0}$ is reported.
  }
  \label{tab:num_bounds}
\end{table}

\subsection{Data Application}

Patients receiving therapies involving central venous catheters (CVCs)
through home infusion agencies may develop CLABSI.
Though relatively rare, CLABSIs are potentially dangerous.
Knowing true CLABSI rates is essential in deploying and testing the impact of CLABSI prevention activities.
Recorded CLABSI rates undercount true positive cases. This is because adjudicators performing CLABSI surveillance
often lack access to the full information 
required to determine whether a CLABSI has occurred \citep{hannum2022task, hannum2023controlling}.
If the available information do not meet the CLABSI definition criteria, as CLABSIs are relatively rare, the adjudicator typically records the CLABSI status as a presumed negative. 

We will apply our zero inflation correction method to data on patients undergoing CVC therapies and thus potentially susceptible to a CLABSI.
Our data contains 652 unique patient records obtained from five different home infusion agencies
across 14 states and the District of Columbia, see \citep{Keller.Hannum.ea.2023.ImplementingValidating} for additional details.
These records correspond to records investigated on patients who presented to a hospital due to a complication and on whom blood cultures were drawn and were positive.
Many patients with CVCs who presented to the hospital due to a complication on whom blood cultures were drawn and were positive do have CLABSIs.
In fact, the observed CLABSI rate in our data was more than $65\%$, much higher than the prevalence in the population undergoing CVC therapies.
However due to zero inflation, even the elevated observed CLABSI rate undercounts the true CLABSI rate in this cohort.

Variables in our data included covariates $C$, which indicated home infusion therapy type and CVC type, coded as binary variables.
A description of these covariates is found in the Appendix.
The outcome of interest is the true CLABSI probability (had zero inflation not occurred), which we denote by $X^{(1)}$.
This outcome is not directly observed. Instead, our data contains the observed CLABSI status $X$, recorded as $0$ and $1$.
Given this variable, we define the inflation indicator $R$ which corresponds to the adjudicator having enough information to make a CLABSI determination for a particular case.
The information could come from private meeting with patients and specialists,
or from reading patients test results and other data in health record systems.
Recording conventions dictate that this indicator has a known value whenever the observed CLABSI is $1$, and is unobserved otherwise (since we cannot distinguish true negatives from inflated zeros).
We considered two candidates for the proxy $W$:
(i) adjudicator access to the shared electronic health record system EPIC, (ii) either adjudicator access to EPIC, or the statewide health information exchange CRISP.
Since $R$ encodes the state of knowing all required information from all sources, we have $R \rightarrow W$.

Our working model is the proxy-augmented ZI MAR under assumptions {\bf A1$^*$} and {\bf A2$^*$}, shown in Fig.~\ref{fig:zi_mcar_mar} (b).
Using the analytic bounds for the ZI MAR model derived in Section~\ref{subsec:pid_zi_mar},
we perform a sensitivity analysis to understand how the true CLABSI rate $p(X^{(1)}=1)$ changes as the proxy-indicator relationship $p(W \mid R)$
varies, within its compatibility range.
First, we use the EM algorithm \citep{dempster77maximum} to maximize the observed data likelihood $\hat{\cal L}_{EM}(X, W, C)$ defined via the full data distribution consistent with our
assumptions.
Next, we invoke Theorem~\ref{thm:zi_mar_bound_2} to obtain the plug-in estimate for $p(W=0 | R=1)$ and the bounds for $p(W=0 | R=0)$.
Finally, we do a grid search over the bounds interval, compute the full data distribution $p(R, X, W, C)$ for each value of $p(W=0 | R=0)$ via
(\ref{eq:kp_id}) and obtain $p(X^{(1)})$ using standard g-formula adjustment in MAR models.
The sensitivity analysis curve is shown in Fig~\ref{fig:clabsi}.
\footnote{
  This plot differs somewhat from the plot in the published version of the paper, due to a corrected data processing error.
  The conclusions on the underlying CLABSI rate were not substantially affected.
}

The values of $p(W=0 | R=0)$ consistent with the model show that inability to make a CLABSI determination is strongly associated with access to patient data via electronic health records.
For proxy EPIC, our obtained (sharp) bound for the nuisance parameter $p(W=0 | R=0)$ is $[0.88, 1]$,
yielding the estimated range of the true CLABSI rate to be $[0.69,0.79]$.
Compared with the baseline rate of $65\%$ under no-zero-inflation assumption, the rate's bound implies that anywhere from $4\%$ to $14\%$ of true CLABSI cases are undercounted, even in our patient cohort with a highly elevated CLABSI prevalence.

We have repeated the analysis using the proxy-augmented
ZI MAR model under assumptions {\bf A1$^{\dag}$} and {\bf A2$^{\dag}$}, shown in Fig.~\ref{fig:zi_mcar_mar} (c).
In this case, bounds for $p(W=0 | R=0, c)$ were obtained, for each value $c$. The narrowest bound $[0.985, 1.0]$ corresponds to adult patients receiving 
outpatient parenteral antimicrobial therapy (OPAT)
via a peripherally inserted central catheter (PICC).
On the other hand, the widest bound $[0.03, 1.0]$ corresponds to pediatric patients receiving chemotherapy via tunneled CVC, a type of catheter under the skin.
We performed a search of $10^6$ points over the polytope comprised of these bounds and find the estimated range of true CLABSI rate to be $[0.77, 0.93]$. That is, anywhere from 12\%-28\% of true CLABSI cases are undercounted.

All derived bounds for the true CLABSI rate were deemed to be medically plausible by our medical collaborators.

\vspace{-0.4cm}
\begin{figure}[ht]
  \centering
  \includegraphics[width=0.5\textwidth]{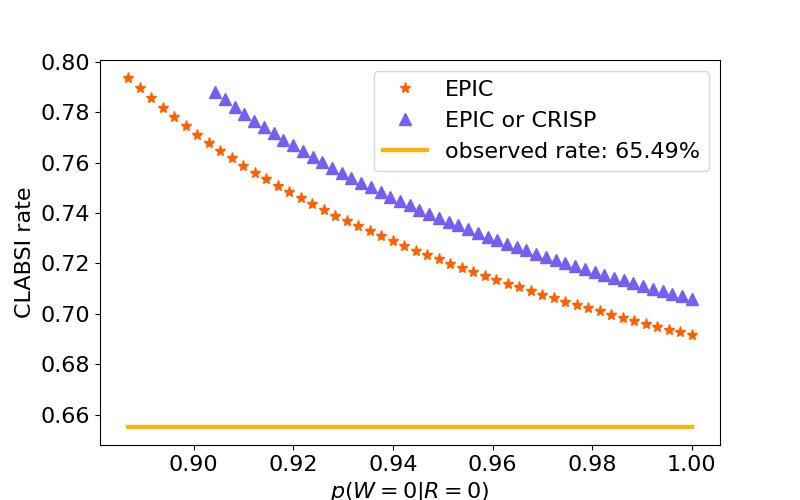}
  \caption{CLABSI rate consistent with model compatible distributions $p(W \mid R)$ under the ZI MAR model with assumptions {\bf A1$^*$}, and {\bf A2$^*$}.
  }
  \label{fig:clabsi}
\end{figure}

\section{Conclusion}

In this paper, we considered inference on data with inflated zeros as
a missing data problem where censored realizations are indicated by a $0$ rather than by a special token such as $``?"$.  This leads to a situation where the censoring indicator for a variable is unobserved any time the value $0$ is observed for such a variable.
We have shown that this significantly complicates the problem, and results in lack of identification even in simple missing data models such as MCAR.

To address this, we proposed a generalization of the approach in \citet{kuroki14measurement} which assumes the existence of an informative proxy for the censoring indicator.  We show that only some relationships between this proxy and the indicator are compatible with the model, derive analytic bounds for this relationship in a number of cases, and show that in some cases our bound is sharp.
Our bounds directly imply bounds on the zero inflated mean parameter.
We verified our results by deriving bounds numerically using the \texttt{autobounds} package described in \citet{duarte23automated}.
Finally, we applied our methods to CLABSI data, which exhibits significant zero inflation.  Our methods led to informative bounds on the true CLABSI rate, and provided a natural sensitivity analysis strategy.

Zero inflation is common in many types of data, particularly in electronic health records.  Our approach provides a principled strategy for deriving informative conclusions from such data without reliance on unrealistic modeling assumptions.

\begin{acknowledgements} 
This research is funded in part by ONR N00014-21-1-2820, NSF 2040804, NSF CAREER 1942239, NIH R01 AI127271-01A1, AHRQ R01 HS027819.
\end{acknowledgements}

\bibliography{references}

\newpage

\onecolumn

\title{Supplementary Material}
\maketitle


\appendix

\section{Proofs}

\subsection{Downstream identification}

\begin{propa}{\ref{prop:full-law-id}}
The full law $p\left(X^{(1)}, R, W, C, X\right)$ exhibiting zero inflation that is Markov relative to an m-DAG $\mathcal{G}$ is identified given $p(R, X, W, C)$ if and only if $\mathcal{G}$ does not contain edges of the form $X_i^{(1)} \rightarrow R_i$ (no self-censoring) and structures of the form $X_j^{(1)} \rightarrow R_i \leftarrow R_j$ (no colluders), and the positivity assumption holds. Moreover, the identifying functional for the full data law coincides with the functional given in \citet{malinsky2021semiparametric}.
\end{propa}
\begin{proof}
    Following the proof in \citet{bhattacharya20completeness}, the full law factorizes as
    \begin{equation}
        \begin{aligned}
            & p\left(R \mid X^{(1)}, C, W\right) \\
            & \quad = \frac{1}{Z} \times \prod_{k=1}^K p\left(R_k \mid R_{-k}=1, X^{(1)}, C, W\right) \\
            & \quad \times \prod_{R_k, R_l \in R} \operatorname{OR}\left(R_k, R_l \mid R_{-(k, l)}=1, X^{(1)}, C, W\right) \\
            & \quad \times \prod_{R_k, R_l, R_m \in R} f\left(R_k, R_l, R_m \mid R_{-(k, l, m)}=1, X^{(1)}, C, W\right) \\
            & \quad \times \prod_{R_k, R_l, R_m, R_n \in R} f\left(R_k, R_l, R_m, R_n \mid R_{-(k, l, m, n)}=1, X^{(1)}, C, W\right) \times \cdots \times f\left(R_1, \ldots, R_K \mid X^{(1)}, C, W\right),
        \end{aligned}
    \end{equation}
    where $R_{-k} = R \setminus \{R_k\}$, and similarly for $X^{(1)}$.
    \begin{itemize}
        \item No-colluder condition implies $R_{k} \Perp X^{(1)}_k \mid R_{-k}, X^{(1)}_{-k}, C, W$, so  $p\left(R_k \mid R_{-k}=1, X^{(1)}, C, W\right) = p\left(R_k \mid R_{-k}=1, X^{(1)}_{-k}, C, W\right)$. Hence these factors use only $R=1$ case of consistency.
        \item The 2-way odd-ratio $\operatorname{OR}\left(R_k, R_l \mid R_{-(k, l)}=1, X^{(1)}, C, W\right)$ is not a function of $\{X^{(1)}_{k}, X^{(1)}_l\}$. Therefore, only case $R=1$ of consistency is used.
        \item the 3-way interaction term $f\left(R_k, R_l, R_m, R_n \mid R_{-(k, l, m, n)}=1, X^{(1)}, C, W\right)$ is not a function of $\{X^{(1)}_{k}, X^{(1)}_l, X^{(1)}_m\}$. Therefore, only case $R=1$ of consistency is used. Similarly for any k-way interaction term.
    \end{itemize}
    Hence the proof in \citet{bhattacharya20completeness} applies to ZI problems, whose consistency differs missing data consistency only at $R=0$ case.
\end{proof}

\subsection{Non-identifiability proof}

\begin{lma}{\ref{lm:nonid}}
Given a ZI model associated with any m-DAG ${\cal G}$, both the target law $p(X^{(1)})$ and the full law $p(X^{(1)}, R, C)$ are non-parametrically non-identified.
\end{lma}

\begin{proof}

Let $\mathcal{G}$ be an m-DAG over $X^{(1)}, R, X, C$ and $\mathcal{P}$ its associated ZI-model. The m-DAG $\mathcal{G}_{\text{mcar}}$ obtained from $\mathcal{G}$ by deleting all edges while keeping $X^{(1)} \rightarrow X \leftarrow R$ defines a sub-model $\mathcal{P}_{\text{mcar}} \subseteq \mathcal{P}$ in which $X^{(1)}, R, C$ are jointly independent. If $p(X^{(1)})$ and $p(X^{(1)}, R, C)$ are non-parametrically non-identified in this sub-model, they are also non-identified in $\mathcal{P}$.

It suffices to prove non-identification for binary variables. The target is $p(X^{(1)}=1)$, and the observed marginals are
\begin{equation}
\begin{aligned}
p(X=1) p(c) &= p(X^{(1)}=1) p(R=1) p(c) \\
p(X=0) p(c) &= p(X^{(1)}=0) p(R=1)p(c) + p(R=0)p(c),
\end{aligned}
\end{equation}
using d-separation in $\mathcal{G}_{\text{mcar}}$. Since the second equation is just $p(c)$ minus the first, if the quantity
\begin{equation}
p(X=1) = p(X^{(1)}=1) p(R=1)
\end{equation}
is shown to be identical for 2 joint distributions in $\mathcal{P}_{\text{mcar}}$, the proof is finished.
Indeed, for any $p_1 \in \mathcal{P}_{\text{mcar}}$, we pick any real number $1 > m  \geq \max \{p_1(X^{(1)}=1), p_1(R=1)\}$ and construct $p_2 \in \mathcal{P}_{\text{mcar}}$ as follow
\begin{equation}
\begin{aligned}
p_2(X^{(1)}=1) &= \frac{1}{m} p_1(X^{(1)}=1); & p_2(R=1) &= m p_1(R=1); & p_2(C) = p_1(C).
\end{aligned}
\end{equation}
Evidently, the target laws are different $p_1(X^{(1)}) \neq p_2(X^{(1)})$, yet the observed marginals are the same $p_1(X, C) = p_2(X, C)$. Moreover, the full laws are also different
\begin{equation}
\begin{aligned}
p_2(X^{(1)} = 0) p_2(R = 1)
&= \left( 1 - \frac{1}{m} p_1(X^{(1)} = 1) \right) m p_1(R=1) \\
&\neq p_1(R=1) - p_1(X^{(1)} = 1)p_1(R=1) \\
&= p_1(X^{(1)} = 0) p_1(R = 1).
\end{aligned}
\end{equation}
Hence, $p(X^{(1)})$ and $p(X^{(1)}, R,C)$ are non-parametrically non-identified in $\mathcal{P}_{\text{mcar}}$.
\end{proof}

\newpage

\subsection{Examples of Compatibility Issue}

Consider the proxy-augmented ZI MCAR model, in which a joint distribution factorizes as

\begin{equation}
p(X^{(1)}, R, X, W) =  p(X^{(1)}, R, X) p(W \mid R).
\end{equation}

Here, the proxy assumptions insist that  
$p(W = 0 \mid R = 0) \neq p(W = 0 \mid R = 1)$.  
Therefore, any $p(W \mid R)$ obeys this inequality is said to be model compatible.  
Moreover, any joint distribution with $p(W \mid R)$ violating this inequality is outside  
of the model. Works investigate marginal models of hidden variable models often  
consider this type of compatibility.

In our paper, we mentioned another type of compatibility. Any joint distribution in the  
model yields a pair of observed law and proxy-indicator conditional distribution  
$(p(X, W), p(W \mid R))$. Obviously, both $p(X, W)$ and $p(W \mid R)$ produced  
this way are model compatible. Furthermore, they are compatible to one another,  
in the sense that there exists a model compatible joint distribution producing them.  
It is possible to construct an incompatible pair $(p(X, W), p(W \mid R))$  
whose components are both model compatible, because the joint distribution yielding  
them is not in the model. This is illustrated in the following simple examples.

\textbf{Example 1:}  

\begin{equation}
\underbrace{
    \begin{array}{c|c|c}
         &  X=0    &  X=1   \\
    \hline
    W=0  &  a      &  b     \\
    W=1  &  c      &  d   
    \end{array}
}_{p(W, X)}
\quad \quad \quad
\underbrace{
    \begin{array}{c|c|c}
         &  R=0  &  R=1  \\
    \hline
    W=0  &  1    &  0    \\
    W=1  &  0    &  1
    \end{array}
}_{p(W \mid R)}
\end{equation}

Since ZI MCAR does not impose any restriction on $p(W, X)$ in the binary case  
(see our proof for the bound in the ZI MCAR case), we can pick any number for  
$a,b,c,d$. In particular, let they be all non-zero. Then both $p(W, X)$ and  
$p(W \mid R)$ are model compatible. However, there isn't any valid $p(R, X)$  
(non-negative, summed to $1$) such that the Kuroki-Pearl equation holds  
$\mathbf{p}_{WX} = \mathbf{p}_{W \mid R}\mathbf{p}_{RX}$.  
Attempting to invert $\mathbf{p}_{W \mid R}$ in this equation will yield  
negative-valued $p(R, X)$.

\textbf{Example 2:}  

We choose a joint distribution (DGP) $p(X^{(1)}, R, X, W)$ Markov to the  
proxy-augmented ZI MCAR graph in Figure 2(a), from which we obtain the true  
$p(W, X)$, true $p(W \mid R)$, true $p(R, X)$.

We calculate $\hat{p}_1(R, X)$ via the matrix inversion equation using the true  
$p(W, X)$ and the true $p(W \mid R)$. The calculated $\hat{p}_1(R, X)$ is valid,  
and close to the true $p(R, X)$ up to floating point precision. This indicates the true  
$p(W, X)$ and the true $p(W \mid R)$ are compatible to one another.

We sample $100000$ data points $(W_i, X_i)$ from this DGP and estimate  
$\hat{p}(W, X)$ by counting, which is the MLE for binary data. Again, this estimation  
is in the model, since marginal model for $p(W, X)$ is saturated in the binary case.  
Then, we calculate $\hat{p}_2(R, X)$ via the matrix inversion equation, using the  
estimated $\hat{p}(W, X)$ and the true $p(W \mid R)$. This estimated  
$\hat{p}_2(R, X)$ has a negative value, which renders it invalid.

The code for this experiment could be found in the supplement of the paper.  
Its output is printed below.

\begin{verbatim}
True p(W,X):
 [[0.42643891  0.31215362]
  [0.14620603  0.11520144]]
True p(W|R):
 [[0.74919143  0.73043156]
  [0.25080857  0.26956844]]
True p(R,X):
 [[0.43502295  0.        ]
  [0.13762199  0.42735506]]
Computed p(R,X) via matrix inv using true p(W,X) and true p(W|R):
 [[0.43502294 -1.81411279e-16]
  [0.13762199  0.42735505]]

Estimated p(W,X):
 [[0.42883     0.30976]
  [0.14496     0.11645]]
Computed p(R,X) via matrix inv using estimated p(W,X) and true p(W|R):
 [[ 0.5178968 -0.08300896]
 [ 0.05589317  0.50921896]]
\end{verbatim}

\newpage

\subsection{ZI MCAR model and bounds}

In this section, $X^{(1)}$ and $C$ are categorical, while $R$ and $W$ are binary.

\subsubsection{Model definition}

Both the ZI MCAR model and ZI MAR model are Cartesian products, between $p(W \mid R)$ model and $p(X^{(1)}, R, X)$ model,
or $p(X^{(1)}, R, X, C)$ model, respectively.
Firstly, the adjustment formula establishes a 1-to-1 relation between the $p(X^{(1)}, R, X, C)$ model and the $p(R, X, C)$ model.
The constraint of the latter is fully understood.

\begin{lemma}
\label{lm:zi_marginal_constraints}C
For 1 variable ZI MCAR and ZI MAR model, the full law model for $p(X^{(1)}, R, X, C)$ is 1-to-1 to the model for $p(R, X, C)$ satisfying \textbf{Z}: $\forall x \neq 0, \forall c, p(X=x, R=0, C=c) = 0$.
\end{lemma}

\begin{proof}
  We only need to prove the lemma for ZI MAR model. 
  \begin{itemize}
    \item $\mathcal{P}$ includes all full laws $p(X^{(1)}, R, X, C)$ factorizing as
      \begin{equation}
          p(X^{(1)}, R, X, C) = p(X \mid X^{(1)}, R) p(X^{(1)} \mid C) p(R \mid C) p(C),
      \end{equation}
      with $p(X \mid X^{(1)}, R)$ denotes the deterministic ZI-consistency.
    
    \item $\mathcal{Q}$ includes all laws $p(R, X, C)$ factorizing as
      \begin{equation}
          p(R, X, C) = p(X \mid R, C) p(R \mid C) p(C).
      \end{equation}
      and obeying \textbf{Z}: $\forall c, \forall x \neq 0: p(X=x, R=0, C=c) = 0$.
  \end{itemize}

  These 2 models are 1-to-1:
  \begin{itemize}
    \item ($ \mathcal{P} \mapsto \mathcal{Q} $): This is just summation $p(R, X, C) = \sum_{X^{(1)}} p(X^{(1)}, R, X, C)$.
      The ZI-consistency implies \textbf{Z}.
    \item ($ \mathcal{Q} \mapsto \mathcal{P} $): By d-separation $p(X^{(1)} \mid C) = p(X \mid R=1, C)$.
  \end{itemize}
\end{proof}

In principle, asking whether $p(W \mid R)$ and $p(X, W, C)$ are compatible means
pointing out a full law $p(X^{(1)}, R, X, C) p(W \mid R)$ which yields both of them.
The above lemma allows us to reformulate this compatibility question by pointing out a joint $p(R, X, C) p(W \mid R)$ in the model.
This has the advantage of simplyfying the original compatibility question, i.e., the polynomial program describing it is of higher degree.
Moreover, we do not sacrify bound sharpness as we invoke this lemma, since the joint $p(R, X, C) p(W \mid R)$ satisfying \textbf{Z} is 1-to-1 to the full law.

\begin{lemma}
  The ZI MCAR model with categorical $X$ and binary $R, W$, which is Markov to the proxy-augmented 
  Fig.~\ref{fig:zi_mcar_mar} (a) (reproduced in Fig.~\ref{fig:zi_mcar_apx}) under proxy assumptions {\bf A1}, {\bf A2}, is described by
  \begin{equation}
  \begin{aligned}
    \mathcal{P} = \left\{
      \begin{array}{c|c}
        (\mathbf{p}_{W \mid R}, \mathbf{p}_{RX})
        &
        \begin{aligned}
          & & & \textstyle \mathbf{p}_{W \mid R} \geq 0, \: \forall r \left(\sum_w p_{w \mid r} = 1 \right), \: p_{w_0 \mid r_0} \neq p_{w_0 \mid r_1}, \\
          & & & \textstyle \mathbf{p}_{RX} \geq 0, \: \sum_{rx} p_{rx} = 1, \: \forall x \neq 0 (p_{r_0 x} = 0) \\
        \end{aligned}
      \end{array}
    \right\}.
  \end{aligned}
  \end{equation}
\end{lemma}

\begin{proof}
Due to $\bf A1$,
\begin{equation}
\begin{aligned}
    p(X^{(1)}, R, X, W) &= p(X^{(1)}, R, X) p(W \mid R).
\end{aligned}
\end{equation}
Therefore, model for $p(X^{(1)}, R, X, W)$ is a Cartesian product between the model for $p(X^{(1)}, R, X)$ and the model for $p(W \mid R)$.
The former is shown to be 1-to-1 to the model for $p(R, X)$ with restriction \textbf{Z}, by lemma~\ref{lm:zi_marginal_constraints}.
\begin{equation}
    \left\{\mathbf{p}_{R, X} \mid \mathbf{p}_{R, X} \geq 0, \sum_{x,r} p_{R, X} = 1, \forall x \neq 0 (p_{r_0, x} = 0) \right\}.
\end{equation}
While the latter is
\begin{equation}
  \left\{\mathbf{p}_{W \mid R} \mid \mathbf{p}_{W \mid R} \geq 0, \forall r \left(\sum_w p_{w \mid r} = 1 \right), \det \mathbf{p}_{W \mid R} \neq 0 \right\}.
\end{equation}
We just need to rewrite $\det \mathbf{p}_{W \mid R} \neq 0$. Since $W, R$ are binary
\begin{equation}
\begin{aligned}
\det \mathbf{p}_{W \mid R} &= p_{w_0 \mid r_0}p_{w_1 \mid r_1} - p_{w_0 \mid r_1} p_{w_1 \mid r_0} \\
&= p_{w_0 \mid r_0} (1 - p_{w_0 \mid r_1}) - p_{w_0 \mid r_1} (1 - p_{w_0 \mid r_0}) \\
&= p_{w_0 \mid r_0} - p_{w_0 \mid r_1} \neq 0.
\end{aligned}
\end{equation}
\end{proof}

\subsubsection{Bounds for ZI MCAR}

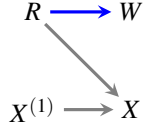
\begin{figure}[!htb]
  \centering
    \begin{tikzpicture}[>=stealth, node distance=1.3cm]
      \tikzstyle{obs} = [circle, inner sep=1pt]
      \tikzstyle{unobs} = [circle, inner sep=0pt]
      \begin{scope}[xshift=0cm]
        \path[->, very thick]
        node[unobs, inner sep=1pt] (x1) {$X^{(1)}$}
        node[obs, right of=x1, inner sep=1pt] (x) {$X$}
        node[unobs, above of=x1, inner sep=1pt] (r) {$R$}
        node[obs, right of=r, inner sep=1pt] (w) {$W$}

        (x1) edge[gray] (x)
        (r) edge[gray] (x)
        (r) edge[blue] (w)
        ;
      \end{scope}
    \end{tikzpicture}
    \caption{
      The graph considered in Theorem~\ref{thm:zi_mcar_bound}: proxy-augmented ZI MCAR model satisfying \textbf{A1} and \textbf{A2} (Fig.~\ref{fig:zi_mcar_mar} a in the main paper).
    }
    \label{fig:zi_mcar_apx}
\end{figure}

Before proving the bound theorem, we have the following useful lemma:
\begin{lemma}
  \label{lm:id_mcar}
  For the ZI MCAR model in Theorem~\ref{thm:zi_mcar_bound}, \textbf{Z} constraint $\forall x \neq 0 (q_{r_0 x}=0)$ is equivalent to $\forall x \neq 0, (q_{w_0 \mid r_1} = q_{w_0 \mid x})$.
  This means: (i) there is a marginal constraint $\forall x \neq 0 (q_{w_0 \mid x} =  q_{w_0 \mid x_1})$, and (ii) $q_{w_0 \mid r_1}$ is point-identified.
\end{lemma}

\begin{proof}

($\Rightarrow$) direction: Suppose $\forall x \neq 0 (q_{r_0 x}=0)$. Then for all $x \neq 0$
\begin{equation}
  \begin{aligned}
    q_{r_0 x} &= 0 & &\Leftrightarrow & q_{x} &= q_{r_0, x} + q_{r_1, x} = q_{r_1, x} & &\Leftrightarrow & q_{r_1 \mid x} &= 1.
  \end{aligned}
\end{equation}
Then, for all $x \neq 0$
\begin{equation}
  \begin{aligned}
    q_{w_0, x} &= q_{w_0 \mid r_0} q_{r_0, x} + q_{w_0 \mid r_1} q_{r_1, x}
                = 0 + q_{w_0 \mid r_1} q_{r_1, x} \\
    \Rightarrow
    q_{w_0 \mid x} &= q_{w_0 \mid r_1} q_{r_1 \mid x} = q_{w_0 \mid r_1}.
  \end{aligned}
\end{equation}

($\Leftarrow$) direction: Suppose $\forall x \neq 0 (q_{w_0 \mid r_1} = q_{w_0 \mid x})$.
Then, for all $x \neq 0$
\begin{equation}
  \begin{aligned}
    q_{w_0, x} &= q_{w_0 \mid r_0} q_{r_0, x} + q_{w_0 \mid r_1} q_{r_1, x} \\
    \Rightarrow
    q_{w_0 \mid x} &= q_{w_0 \mid r_0} q_{r_0 \mid x} + q_{w_0 \mid r_1} q_{r_1 \mid x} \\
    \Rightarrow
    q_{w_0 \mid r_1} &= q_{w_0 \mid r_0} q_{r_0 \mid x} + q_{w_0 \mid r_1} q_{r_1 \mid x} \\
    \Rightarrow
                0 &= (q_{w_0 \mid r_0} - q_{w_0 \mid r_1}) q_{r_0 \mid x}. \\
  \end{aligned}
\end{equation}
Since $q_{w_0 \mid r_0} \neq q_{w_0 \mid r_1}$, we must have $q_{r_0 \mid x} = 0 \Rightarrow q_{r_0, x} = 0$. This is true for all $x \neq 0$.

\end{proof}

Due to this lemma, an observed law $p(X, W)$ is consistent to the model if and only if $\forall x \neq 0, p(W=0, X=x) = p(W=0, X=1)$.
We also require positivity $\forall x, p(X=x) > 0$, so that $p(W \mid X)$ is well-defined.

\begin{thma}{\ref{thm:zi_mcar_bound}}
  Consider a ZI MCAR model in Fig.~\ref{fig:zi_mcar_mar} (a) (reproduced in Fig.~\ref{fig:zi_mcar_apx}) under proxy assumptions {\bf A1}, {\bf A2},
  with categorical $X$ and binary $R, W$.
  Given a consistent observed law $p(X, W)$ satisfying positivity assumption, $\forall x, p(x) > 0$,
  the set of compatible proxy-indicator conditionals $q(W \mid R)$ is given by
  \begin{align*}
    q_{w_0 \mid r_1} &= p_{w_0 \mid x_1}\\
    q_{w_0 \mid r_0} &\in
    \begin{cases}
    [ p_{w_0 \mid x_0}, 1] \text{ if } p_{w_0 \mid x_0} > p_{w_0 \mid x_1}\\
    [ 0, p_{w_0 \mid x_0}] \text{ if } p_{w_0 \mid x_0} < p_{w_0 \mid x_1}\\
    (0, 1) \setminus \{ p_{w_0 \mid x_0} \}\text{ if } p_{w_0 \mid x_0} = p_{w_0 \mid x_1}
    \end{cases}
  \end{align*}
  These bounds are sharp.
  Moreover, if $p_{w_0 \mid x_0} = p_{w_0 \mid x_1}$, $p(X, W)$ must satisfy $0 < p_{w_0 \mid x_0} < 1$,
  and zero inflation does not occur, i.e., $q(R=0)=0$.
\end{thma}

\begin{proof}

In the following, $q(\cdot)$ denotes an element in a model, while $p(\cdot)$ is derived from the given marginal $p(X, W)$.

In principle, any compatible $q(W \mid R)$ to $p(X, W)$ must be derived from some full joint distribution $q(X^{(1)}, R, X) q(W \mid R)$,
such that $q(X, W) = p(X, W)$.
Since the model for $q(X^{(1)}, R, X)$ is 1-to-1 to the model for $q(R, X)$ with restriction $\textbf{Z}$,
we can simplify this process by considering the marginal model for $q(R, X, W)$
\begin{equation}
\begin{aligned}
  \mathcal{P} = \left\{
    \begin{array}{c|c}
      (\mathbf{q}_{W \mid R}, \mathbf{q}_{RX})
      &
      \begin{aligned}
        & & & \textstyle \mathbf{q}_{W \mid R} \geq 0, \: \forall r \left(\sum_w q_{w \mid r} = 1 \right), \: q_{w_0 \mid r_0} \neq q_{w_0 \mid r_1}, \\
        & & & \textstyle \mathbf{q}_{RX} \geq 0, \: \sum_{rx} q_{rx} = 1, \: \forall x \neq 0 (q_{r_0 x} = 0) \\
      \end{aligned}
    \end{array}
  \right\}.
\end{aligned}
\end{equation}
The subset of $q(R, X, W)$ yielding the observed law $p(X, W)$ is
\begin{equation}
\begin{aligned}
\mathcal{Q}
&= \left\{(\mathbf{q}_{W \mid R}, \mathbf{q}_{RX}) \mid \: \mathbf{q}_{W \mid R} \mathbf{q}_{RX} = \mathbf{p}_{WX}, (\mathbf{q}_{W \mid R}, \mathbf{q}_{RX}) \in \mathcal{P} \right\}.
\end{aligned}
\end{equation}

\textbf{Polynomial program}

Since $\mathbf{q}_{W \mid R} \in \mathcal{P}$ is invertible, $\mathcal{Q}$ is the set of all pairs $(\mathbf{q}_{W \mid R}, \mathbf{q}_{RX})$ with $\mathbf{q}_{RX} = [\mathbf{q}_{W \mid R}]^{-1} \mathbf{p}_{WX}$
and $\mathbf{q}_{W \mid R} \in \mathcal{B}$,
\begin{equation}
\begin{aligned}
\mathcal{B}
  &=
  \left\{
    \begin{array}{c|c}
      \mathbf{q}_{W \mid R}
      &
      \begin{aligned}
        & & &\textstyle \mathbf{q}_{W \mid R} \geq 0, \: \forall r \left(\sum_w q_{w \mid r} = 1 \right), \: q_{w_0 \mid r_0} \neq q_{w_0 \mid r_1}, \\
        & & &\textstyle \mathbf{q}_{RX} \geq 0, \: \sum_{rx} q_{rx} = 1, \: \forall x \neq 0 (q_{r_0 x} = 0), \\
        & & &\textstyle \text{where } \mathbf{q}_{RX} = [\mathbf{q}_{W \mid R}]^{-1} \mathbf{p}_{WX}. \\
      \end{aligned}
    \end{array}
  \right\}.
\end{aligned}
\end{equation}
$\mathcal{B}$ is called the compatibility set of $q(W \mid R)$ w.r.t. $p(X, W)$.
As mentioned in the main paper, one can directly solve for $\mathcal{B}$ via the following polynomial program, where $\mathbf{q}_{RX}$ are slack variables.
\begin{equation}
  \begin{aligned}
    \max_{ q_{w_0 \mid r_0}} &&& \pm q_{w_0 \mid r_0}
    \\
    \text{s.t.}
    &&& \mathbf{q}_{W \mid R} \mathbf{q}_{RX} = \mathbf{p}_{WX}, \\
    &&& \textstyle \mathbf{q}_{W \mid R} \geq 0, \: \forall r (\sum_{w} q_{w|r} = 1), \: q_{w_0 \mid r_0} \neq q_{w_0 \mid r_1}, \\
    &&& \textstyle \mathbf{q}_{RX} \geq 0, \: \sum_{rx} q_{rx} = 1, \: \forall x \neq 0 (q_{r_0 x} = 0).
  \end{aligned}
\end{equation}
As we will show below, the constraint $\forall x \neq 0 (q_{r_0 x} = 0)$ is equivalent to $q_{w_0 \mid r_1} = p_{w_0 \mid x_1}$.
This is a quadratic program due to the first constraint.

\textbf{Linear program}

We will simplify $\mathcal{B}$. We do so by considering its superset and adding constraints to it.
Firstly, $\mathcal{B}$ could be parameterized by only 2 numbers, because its superset is
\begin{equation}
\begin{aligned}
  \left\{
    \begin{array}{c|c}
      \mathbf{q}_{W \mid R}
      &
      \mathbf{q}_{W \mid R} \geq 0, \: \forall r \left(\sum_w q_{w \mid r} = 1 \right)
    \end{array}
  \right\}
  =
  \left\{
    \begin{array}{c|c}
      \textstyle \mathbf{q}_{W \mid R} = \begin{pmatrix} q_{w_0 \mid r_0} & q_{w_0 \mid r_1} \\ 1-q_{w_0 \mid r_0} & 1-q_{w_0 \mid r_1} \end{pmatrix}
      &
      0 \leq q_{w_0 \mid r_0} \leq 1, \: 0 \leq q_{w_0 \mid r_1} \leq 1
    \end{array}
  \right\}.
\end{aligned}
\end{equation}

Secondly, when $q_{w_0 \mid r_0} \neq q_{w_0 \mid r_1}$, the 2-2 matrix $\mathbf{q}_{W \mid R}$ has inverse
\begin{equation}
\left[\mathbf{q}_{W \mid R}\right]^{-1} = \frac{1}{q_{w_0 \mid r_0} - q_{w_0 \mid r_1}} \underbrace{\begin{pmatrix} 1-q_{w_0 \mid r_1} & - q_{w_0 \mid r_1} \\ q_{w_0 \mid r_0}-1 & q_{w_0 \mid r_0} \end{pmatrix}}_{M}.
\end{equation}
Therefore, we can transform the quadratic constraint into the following equivalent linear constraints
\begin{equation}
\begin{aligned}
  &
  \left\{
    \begin{array}{c|c}
      \mathbf{q}_{W \mid R}
      &
      \textstyle
      \mathbf{q}_{W \mid R} \geq 0, \: \forall r \left(\sum_w q_{w \mid r} = 1 \right), \: q_{w_0 \mid r_0} \neq q_{w_0 \mid r_1}, \:
      \mathbf{q}_{RX} = [\mathbf{q}_{W \mid R}]^{-1} \mathbf{p}_{WX} \geq 0 \\
    \end{array}
  \right\}
  \\
  &=
  \bigcup_{s \in \{1,-1\}} \left\{
    \begin{array}{c|c}
      \mathbf{q}_{W \mid R}
      &
      \mathbf{q}_{W \mid R} \geq 0, \: \forall r \left(\sum_w q_{w \mid r} = 1 \right), \: sq_{w_0 \mid r_0} > sq_{w_0 \mid r_1}, \: sM \mathbf{p}_{WX} \geq 0
    \end{array}
  \right\}.
\end{aligned}
\end{equation}

Next, given $\mathbf{q}_{W \mid R} \mathbf{q}_{RX} = \mathbf{p}_{WX}$ where all terms are non-negative, $\sum_w q_{w \mid R} = 1$, $\sum_{w,x} p_{w,x} = 1$,
and $\mathbf{q}_{W \mid R}$ is invertible, then $\sum_{rx} q_{rx} = 1$. Hence, $\sum_{rx} q_{rx} = 1$ is a redundant constraint.
\textit{Proof:} entry $ij$-th $[\mathbf{p}_{WX}]_{ij} = q_{w_i \mid r_0} q_{r_0 x_j} + q_{w_i \mid r_1} q_{r_1 x_j}$.
Then $1 = \sum_{ij} [\mathbf{p}_{WX}]_{ij} = \sum_{j} \left( (\sum_{i} q_{w_i \mid r_0}) q_{r_0 x_j} + (\sum_{i} q_{w_i \mid r_1}) q_{r_1 x_j} \right) = \sum_{j} (q_{r_0 x_j} + q_{r_1 x_j})$.

Finally, lemma~\ref{lm:id_mcar} says $\forall x \neq 0 (q_{r_0 x}=0) \Leftrightarrow \forall x \neq 0 (q_{w_0 \mid r_1} = q_{w_0 \mid x})$,
and $ q_{w_0 \mid x} = p_{w_0 \mid x}$ in $\mathcal{B}$.
Note that this lemma also requires $p(X, W)$ to satisfy the marginal constraint $p_{w_0 \mid x} = p_{w_0 \mid x_1}$.
Therefore, the constraint $\forall x \neq 0 (q_{r_0 x}=0) \Leftrightarrow q_{w_0 \mid r_1} = p_{w_0 \mid x_1}$.

Putting together, we can write $\mathcal{B}$ as
\begin{equation}
\begin{aligned}
\mathcal{B}
=
\bigcup_{s \in \{1,-1\}} \left\{
  \begin{array}{c | c}
    \mathbf{q}_{W \mid R} = \begin{pmatrix} q_{w_0 \mid r_0} & q_{w_0 \mid r_1} \\ 1-q_{w_0 \mid r_0} & 1-q_{w_0 \mid r_1} \end{pmatrix}
    &
    \begin{aligned}
      &&& \textstyle s\begin{pmatrix} 1-q_{w_0 \mid r_1} & - q_{w_0 \mid r_1} \\ q_{w_0 \mid r_0}-1 & q_{w_0 \mid r_0} \end{pmatrix}  \mathbf{p}_{WX}  \geq \mathbf{0}, \\
      &&& \textstyle s \cdot q_{w_0 \mid r_0} > s \cdot q_{w_0 \mid r_1}, 0 \leq q_{w_0 \mid r_0} \leq 1, \\
      &&& \textstyle q_{w_0 \mid r_1} = p_{w_0 \mid x_1}
    \end{aligned}
  \end{array}
\right\}.
\end{aligned}
\end{equation}
or,
\begin{equation}
\begin{aligned}
&\mathcal{B}
=
\left\{
  \begin{array}{c | c}
    \mathbf{q}_{W \mid R} = \begin{pmatrix} q_{w_0 \mid r_0} & q_{w_0 \mid r_1} \\ 1-q_{w_0 \mid r_0} & 1-q_{w_0 \mid r_1} \end{pmatrix}
    &
    q_{w_0 \mid r_1} = p_{w_0 \mid x_1}, \:
    q_{w_0 \mid r_0} \in \mathcal{B}_{w_0 \mid r_0}
  \end{array}
\right\},
\\
&\mathcal{B}_{w_0 \mid r_0}
=
\bigcup_{s \in \{1,-1\}} \mathcal{B}_{w_0 \mid r_0}^s
=
\bigcup_{s \in \{1,-1\}} \left\{
  \begin{array}{c | c}
    q_{w_0 \mid r_0}
    &
    \begin{aligned}
      &&& \textstyle s\begin{pmatrix} 1-q_{w_0 \mid r_1} & - q_{w_0 \mid r_1} \\ q_{w_0 \mid r_0}-1 & q_{w_0 \mid r_0} \end{pmatrix}  \mathbf{p}_{WX}  \geq \mathbf{0}, \\
      &&& \textstyle s \cdot q_{w_0 \mid r_0} > s \cdot q_{w_0 \mid r_1}, \: 0 \leq q_{w_0 \mid r_0} \leq 1, \\
      &&& \textstyle q_{w_0 \mid r_1} = p_{w_0 \mid x_1}
    \end{aligned}
  \end{array}
\right\}.
\end{aligned}
\end{equation}
The set $\mathcal{B}_{w_0 \mid r_0}$ is called the \textbf{compatible set} of $q_{w_0 \mid r_0}$ w.r.t. $p(X, W)$.
As will be shown, this is an interval in $[0,1]$, hence the name {\bf compatibility bound}.

To find $\mathcal{B}_{w_0 \mid r_0}$, we will find each $\mathcal{B}^{s}_{w_0 \mid r_0}$ and take their union.
Each $\mathcal{B}_{w_0 \mid r_0}^s$ could be numerically computed by solving the 2 linear programs
\begin{equation}
  \begin{aligned}
    \max_{ q_{w_0 \mid r_0} } &&& \pm q_{w_0 \mid r_0} \\
    \text{s.t.}
      &&& s \cdot \begin{pmatrix} 1-q_{w_0 \mid r_1} & - q_{w_0 \mid r_1} \\ q_{w_0 \mid r_0}-1 & q_{w_0 \mid r_0} \end{pmatrix}  \mathbf{p}_{WX}  \geq \mathbf{0}, \\
      &&& s \cdot q_{w_0 \mid r_0} > s \cdot q_{w_0 \mid r_1}, \: 0 \leq q_{w_0 \mid r_0} \leq 1, \\
      &&& q_{w_0 \mid r_1} = p_{w_0 \mid x_1}.
  \end{aligned}
\end{equation}
These problems are linear program as $q_{w_0 \mid r_0}$ is the only unknown and all constraints are linear.
The set $\mathcal{B}_{w_0 \mid r_0}^s$ is the interval whose endpoints are 2 numbers returned by these programs.

\textbf{Solutions to linear programs}

\textit{Solving $\mathcal{B}_{w_0 \mid r_0}^{s=1}$:} We expand the matrix multiplication equation
\begin{equation}
\begin{aligned}
                 & & p_{w_0, x_0} \left( 1 - q_{w_0 \mid r_1} \right) - p_{w_1, x_0} q_{w_0 \mid r_1} & \geq 0
                 & &\Leftrightarrow & & p_{w_0 \mid x_0} p_{w_1 \mid x_1} - p_{w_1 \mid x_0} p_{w_0 \mid x_1} \geq 0 \\
  \forall x \neq 0,& & p_{w_0, x} \left( 1 - q_{w_0 \mid r_1} \right) - p_{w_1, x} q_{w_0 \mid r_1} & \geq 0
                 & &\Leftrightarrow &  & \forall x \neq 0, p_{w_0, x} p_{w_1 \mid x_1} - p_{w_1, x} p_{w_0 \mid x_1} = 0 \\
                 & & p_{w_0, x_0} \left( q_{w_0 \mid r_0} - 1 \right) + p_{w_1, x_0} q_{w_0 \mid r_0} & \geq 0
                 & &\Leftrightarrow & & q_{w_0 \mid r_0} \geq \frac{p_{w_0, x_0}}{ p_{w_0, x_0} + p_{w_1, x_0} } = p_{w_0 \mid x_0} \\
  \forall x \neq 0,& & p_{w_0, x} \left( q_{w_0 \mid r_0} - 1 \right) + p_{w_1, x} q_{w_0 \mid r_0} & \geq 0
                 & &\Leftrightarrow & & \forall x \neq 0, q_{w_0 \mid r_0} \geq \frac{p_{w_0, x}}{ p_{w_0, x} + p_{w_1, x} } = p_{w_0 \mid x}. \\
\end{aligned}
\end{equation}

In the derivations above, we use the positivity assumption $\forall x (p(X=x) > 0)$.
At the very least, we assume there is zeros, i.e., $p(X=0)>0$, otherwise the problem does not make sense.
If positivity is violated, e.g., $\exists x \neq 0, p(X=x) = 0$, one can show that $p_{w_0, x} = p_{w_1, x} = 0$,
and hence this value $x$ does not place any restriction on $q_{w_0 \mid r_0}$, and can be ignored in the following discussion.

The first equation shows that the $s=1$ case has no solution if $p_{w_0 \mid x_0} p_{w_1 \mid x_1} - p_{w_1 \mid x_0} p_{w_0 \mid x_1} < 0$.
When the LHS is non-negative, the feasible region $\mathcal{B}_{w_0 \mid r_0}^{s=1}$ is $\max_x p_{w_0 \mid x} \leq q_{w_0 \mid r_0} \leq 1$.
We can further split into 2 cases, and note that $q_{w_0 \mid r_0} > q_{w_0 \mid r_1}$, per $s=1$.
\begin{enumerate}
  \item If $p_{w_0 \mid x_0} p_{w_1 \mid x_1} - p_{w_1 \mid x_0} p_{w_0 \mid x_1} = 0 \Leftrightarrow p_{w_0 \mid x_0} = p_{w_0 \mid x_1}$,
  which is true for all values in $[0,1]$.
  Then $p_{w_0 \mid x_0} < q_{w_0 \mid r_0} \leq 1$. For this to make sense, we must have $p_{w_0 \mid x_0} < 1$.
  \item If $p_{w_0 \mid x_0} p_{w_1 \mid x_1} - p_{w_1 \mid x_0} p_{w_0 \mid x_1} > 0 \Leftrightarrow p_{w_0 \mid x_0} > p_{w_0 \mid x_1}$,
  which is true for all $0 \leq p_{w_0 \mid x_1} < 1$.
  Then $p_{w_0 \mid x_0} \leq q_{w_0 \mid r_0} \leq 1$.
\end{enumerate}
The bounds are sharp because they are the feasible regions $\mathcal{B}_{w_0 \mid r_0}^{s=1}$.

\textit{Solving $\mathcal{B}_{w_0 \mid r_0}^{s=-1}$:} Similarly
\begin{equation}
\begin{aligned}
                   & & p_{w_0, x_0} \left( 1 - q_{w_0 \mid r_1} \right) - p_{w_1, x_0} q_{w_0 \mid r_1} & \leq 0
                   & &\Leftrightarrow & & p_{w_0 \mid x_0} p_{w_1 \mid x_1} - p_{w_1 \mid x_0} p_{w_0 \mid x_1} \leq 0 \\
    \forall x \neq 0,& & p_{w_0, x} \left( 1 - q_{w_0 \mid r_1} \right) - p_{w_1, x} q_{w_0 \mid r_1} & \leq 0
                   & &\Leftrightarrow &  & \forall x \neq 0, p_{w_0, x} p_{w_1 \mid x_1} - p_{w_1, x} p_{w_0 \mid x_1} = 0 \\
                   & & p_{w_0, x_0} \left( q_{w_0 \mid r_0} - 1 \right) + p_{w_1, x_0} q_{w_0 \mid r_0} & \leq 0
                   & &\Leftrightarrow & & q_{w_0 \mid r_0} \leq \frac{p_{w_0, x_0}}{ p_{w_0, x_0} + p_{w_1, x_0} } = p_{w_0 \mid x_0} \\
    \forall x \neq 0,& & p_{w_0, x} \left( q_{w_0 \mid r_0} - 1 \right) + p_{w_1, x} q_{w_0 \mid r_0} & \leq 0
                   & &\Leftrightarrow & & \forall x \neq 0, q_{w_0 \mid r_0} \leq \frac{p_{w_0, x}}{ p_{w_0, x} + p_{w_1, x} } = p_{w_0 \mid x}. \\
\end{aligned}
\end{equation}
  
The first equation shows that the $s=1$ case has no solution if $p_{w_0 \mid x_0} p_{w_1 \mid x_1} - p_{w_1 \mid x_0} p_{w_0 \mid x_1} > 0$.
When the LHS is non-positive, the feasible region $\mathcal{B}_{w_0 \mid r_0}^{s=-1}$ is $0 \leq q_{w_0 \mid r_0} \leq \min_x p_{w_0 \mid x}$.
We can further split into 2 cases, and note that $q_{w_0 \mid r_0} < q_{w_0 \mid r_1} $, per $s=-1$.
\begin{enumerate}
  \item If $p_{w_0 \mid x_0} p_{w_1 \mid x_1} - p_{w_1 \mid x_0} p_{w_0 \mid x_1} = 0 \Leftrightarrow p_{w_0 \mid x_0} = p_{w_0 \mid x_1}$,
  which is true for all values in $[0,1]$.
  Then $0 \leq q_{w_0 \mid r_0} < p_{w_0 \mid x_0}$. For this to make sense, we must have $p_{w_0 \mid x_0} > 0$.
  \item If $p_{w_0 \mid x_0} p_{w_1 \mid x_1} - p_{w_1 \mid x_0} p_{w_0 \mid x_1} < 0 \Leftrightarrow p_{w_0 \mid x_0} < p_{w_0 \mid x_1}$.
  which is true for all $0 < p_{w_0 \mid x_1} \leq 1$.
  Then $0 \leq q_{w_0 \mid r_0} \leq p_{w_0 \mid x_0}$.
        
\end{enumerate}
The bounds are sharp because they are the feasible regions $\mathcal{B}_{w_0 \mid r_0}^{s=-1}$.

\textbf{Result}

Combine these results to get the compatibility bound $\mathcal{B}_{w_0 \mid r_0} = \bigcup_{s \in \{1,-1\}} \mathcal{B}_{w_0 \mid r_0}^s$.
The bounds are sharp.
\begin{align*}
    q_{w_0 \mid r_1} &= p_{w_0 \mid x_1}\\
    q_{w_0 \mid r_0} &\in
    \begin{cases}
    [ p_{w_0 \mid x_0}, 1] \text{ if }p_{w_0 \mid x_0} > p_{w_0 \mid x_1}\\
    [ 0, p_{w_0 \mid x_0}] \text{ if }p_{w_0 \mid x_0} < p_{w_0 \mid x_1}\\
    (0, 1) \setminus \{ p_{w_0 \mid x_0} \}\text{ if } 0 < p_{w_0 \mid x_0} = p_{w_0 \mid x_1} < 1
    \end{cases}
\end{align*}
The situations $p_{w_0 \mid x_0} = p_{w_0 \mid x_1} \in \{0,1\}$ are not allowed by the model.
Moreover, if $p_{w_0 \mid x_0} = p_{w_0 \mid x_1}$ then $q(R=0)=0$, i.e., zero inflation does not occur.
\textit{Proof}:
\begin{equation}
  \begin{aligned}
    p_{w_0 \mid x_0} & = q_{w_0 \mid r_0} q_{r_0 \mid x_0} + q_{w_0 \mid r_1} q_{r_1 \mid  x_0} \\
    p_{w_0 \mid x_1} & = q_{w_0 \mid r_1} \quad (\text{id})
  \end{aligned}
\end{equation}
Therefore, subtracting both sides,
\begin{equation}
  0 = p_{w_0 \mid x_0} - p_{w_0 \mid x_1} = q_{w_0 \mid r_0} q_{r_0 \mid x_0} - q_{w_0 \mid r_1} q_{r_0 \mid x_0} = \left(q_{w_0 \mid r_0}  - q_{w_0 \mid r_1} \right) q_{r_0 \mid x_0}.
\end{equation}
Due to proxy assumption \textbf{A2}: $q_{w_0 \mid r_0} \neq q_{w_0 \mid r_1}$. Then the LHS equals $0$ if and only if $q_{r_0 \mid x_0} = 0$.
Moreover, \textbf{Z} implies $\forall x \neq 0, q_{r_0 \mid x} = 0$. Then $q_{r_0} = \sum_x q_{r_0 \mid x} p_{x} = 0$.

\end{proof}


\pagebreak
\subsection{ZI MAR proofs}


\begin{figure}[!htb]
  \centering
  \begin{tikzpicture}
    \tikzstyle{obs} = [circle, inner sep=1pt]
    \tikzstyle{unobs} = [circle, inner sep=0pt]
    \begin{scope}
      \path[->, very thick]
      node[unobs, inner sep=1pt] (x1) {$X^{(1)}$}
      node[obs, right of=x1, inner sep=1pt] (x) {$X$}
      node[unobs, above of=x1, inner sep=1pt] (r) {$R$}
      node[obs, below of=r, xshift=-0.6cm, yshift=+0.65cm, inner sep=1pt] (c) {$C$}
      node[obs, right of=r, inner sep=1pt] (w) {$W$}

      (x1) edge[gray] (x)
      (r) edge[gray] (x)
      (r) edge[blue] (w)
      (c) edge[blue] (r)
      (c) edge[blue] (x1)
      (c) edge[blue] (w)
      ;
    \end{scope}
  \end{tikzpicture}
  \caption{
    The graph considered in Theorem~\ref{thm:zi_mar_bound_1}: proxy-augmented ZI MAR model satisfying \textbf{A1}$^{\dag}$ and \textbf{A2}$^{\dag}$ (Fig.~\ref{fig:zi_mcar_mar} (b) in the main paper).
  }
  \label{fig:zi_mar_apx_1}
\end{figure}

\begin{thma}{\ref{thm:zi_mar_bound_1}}
  Consider a ZI MAR model in Fig.~\ref{fig:zi_mcar_mar} (b) (reproduced in Fig.~\ref{fig:zi_mar_apx_1})
  under proxy assumptions {\bf A1$^{\dag}$} and {\bf A2$^{\dag}$}, 
  with categorical $X, C$ and binary $R, W$.
  Given a consistent observed law $p(X, W, C)$ satisfying positivity assumption, $\forall x, c, p(x,c) > 0$,
  the set of compatible proxy-indicator conditional distributions $q(W \mid R,C)$ is given by, for each  value $c$,
  \begin{align*}
      q_{w_0 \mid r_1, c} &= p_{w_0 \mid x_1, c}\\
      q_{w_0 \mid r_0, c} &\in
      \begin{cases}
      [ p_{w_0 \mid x_0, c}, 1] \text{ if } p_{w_0 \mid x_0, c} > p_{w_0 \mid x_1, c}\\
      [ 0, p_{w_0 \mid x_0, c}] \text{ if } p_{w_0 \mid x_0, c} < p_{w_0 \mid x_1, c}\\
      (0, 1) \setminus \{ p_{w_0, \mid x_0, c} \}\text{ if } p_{w_0 \mid x_0, c} = p_{w_0 \mid x_1, c}
      \end{cases}
  \end{align*}
  These bounds are sharp.
  Moreover, if $p_{w_0 \mid x_0, c} = p_{w_0 \mid x_1, c}$, $p(X, W, C)$ must satisfy $0 < p_{w_0 \mid x_0, c} < 1$,
  and zero inflation does not occur for stratum $C=c$, i.e., $q(R=0 \mid c)=0$.
\end{thma}

\begin{proof}


\item
\textbf{Model definition.}

We assume $C$ is a cardinal variable, taking values in a finite set $\mathcal{C}$. Any joint distribution in this ZI MAR model is
\begin{equation}
q \left( X^{(1)}, R, X, W, C \right) = q \left( X^{(1)}, R, X, W \mid C \right) p(C).
\end{equation}
Since the Markov factors are variationally independent, the ZI MAR model is a Cartesian product
\begin{equation}
\begin{aligned}
\mathcal{P}^{(1)}_{\text{ZI MAR}}
&=
\left( \textstyle \bigotimes_{c \in \mathcal{C}} \mathcal{P}^{(1)}_{\text{ZI MCAR}} (c) \right) \otimes \mathcal{P}_{C}
\\
\mathcal{P}_{C} &=\{q(C)\},
\\
\mathcal{P}^{(1)}_{\text{ZI MCAR}} (c) &= \left\{q(X^{(1)}, R, X, W \mid c) \mid \textbf{A1}, \textbf{A2} \right\} \\
\end{aligned}
\end{equation}
Note how constraints \textbf{A1}$^\dag$, \textbf{A2}$^\dag$ are equivalent to imposing \textbf{A1}, \textbf{A2} to each stratum $C=c$.
Notation: (i) $\mathcal{P}_{C} = \{q(C)\}$ means $\mathcal{P}_{C}$ is a non-parametric model contains all probability distribution $q(C)$, and (ii) probability constraints are assumed to hold.

In this product, $\mathcal{P}^{(1)}_{\text{ZI MCAR}} (c)$ for all $c$ are the same ZI MCAR model described in Theorem~\ref{thm:zi_mcar_bound}, repeated $|\mathcal{C}|$ times.
The value $c$ is not a parameter of the model $\mathcal{P}^{(1)}_{\text{ZI MCAR}} (c)$, but a constant.
Its only purpose is for the sake of book-keeping when constructing the joint distribution in $\mathcal{P}^{(1)}_{\text{ZI MAR}}$.
For MAR, standard adjustment method point identifies $q(X^{(1)}, R, X, W)$ as a functional of $q(R, X, W)$.
Therefore, as shown in lemma \ref{lm:zi_marginal_constraints} the set $\mathcal{P}^{(1)}_{\text{ZI MCAR}} (c)$ is 1-to-1 to the set
$\mathcal{P}_{\text{ZI MCAR}}(c) = \left\{q(R, X, W \mid c) \mid {\bf Z, A1, A2}\right\}$.
Hence, we are interested in the marginal model
\begin{equation}
\begin{aligned}
\mathcal{P}_{\text{ZI MAR}}
&=
\left( \textstyle \bigotimes_{c \in \mathcal{C}} \mathcal{P}_{\text{ZI MCAR}} (c) \right) \otimes \mathcal{P}_{C}. \\
\end{aligned}
\end{equation}

\textbf{Finding compatible set.}

Given an observed law $\mathbf{p}_{WX C}$, we want to find the compatible set w.r.t. this law
\begin{equation}
\begin{aligned}
\mathcal{Q}
&= \left\{ \mathbf{q}_{XRWC} \mid \mathbf{q}_{XRWC} \in \mathcal{P}_{\text{ZI MAR}}, \: \forall c \in \mathcal{C} \left( \mathbf{q}_{W \mid R c} \mathbf{q}_{RX \mid c} = \mathbf{p}_{WX \mid c} \right), \: \forall c \in \mathcal{C} \left( q(c) = p(c) \right) \right\}. \\
\end{aligned}
\end{equation}

Geometrically speaking, this set is the intersection of our model $\mathcal{P}_{\text{ZI MAR}}$ with the constraint set $\mathcal{E}$, which is itself a Cartesian product,
\begin{equation}
\begin{aligned}
\mathcal{E}
&=
\left\{ \mathbf{q}_{XRWC} \mid \forall c \in \mathcal{C} \left( \mathbf{q}_{W \mid R c} \mathbf{q}_{RX \mid c} = \mathbf{p}_{WX \mid c} \right), \text{ and } \forall c \in \mathcal{C} \left( q(c) = p(c) \right) \right\} \\
&=
\textstyle \bigotimes_{c \in \mathcal{C}} \left\{ \mathbf{q}_{XRW \mid c} \mid \mathbf{q}_{W \mid R c} \mathbf{q}_{RX \mid c} = \mathbf{p}_{WX \mid c} \right\} \otimes \left\{ p(c) \right\}. \\
\end{aligned}
\end{equation}
Here we abuse notation $\left\{ p(c) \right\}$ to mean the set with 1 element - the observed law $p(C)$, which is not the model $\mathcal{P}_C$.

Since the constraint $\mathbf{q}_{W \mid R c} \mathbf{q}_{RX \mid c} = \mathbf{p}_{WX \mid c}$ only concerns $q(R, X, W \mid c)$ and does not concern other $q(R, X, W \mid c')$ in any way, we push each constraint to the corresponding $\mathcal{P}_{\text{ZI MCAR}} (c)$. In other words, we will proceed to find the ZI MCAR compatibility bound for each level $c$, as shown below. Mathematically, as Cartesian product could be written as intersection: If $A,C \subseteq U$ and $B,D \subseteq V$, then
\begin{equation}
\begin{aligned}
(A \otimes B) \cap (C \otimes D)
&= (A \otimes V) \cap (U \otimes B) \cap (C \otimes V) \cap (U \otimes D)
\\
&= (A \cap C \otimes V) \cap (U \otimes (B \cap D))
\\
&= (A \cap C) \otimes (B \cap D).
\end{aligned}
\end{equation}

We could transform
\begin{equation}
\begin{aligned}
  \mathcal{Q}
  &=
  \mathcal{P}_{\text{ZI MAR}} \cap \mathcal{E}
  =
  \textstyle \bigg( \left( \bigotimes_{c \in \mathcal{C}} \mathcal{P}_{\text{ZI MCAR}} (c) \right) \otimes \mathcal{P}_{C} \bigg)
  \cap
  \bigg( \bigotimes_{c \in \mathcal{C}} \left\{ \mathbf{q}_{XRW \mid c} \mid \mathbf{q}_{W \mid R c} \mathbf{q}_{RX \mid c} = \mathbf{p}_{WX \mid c} \right\} \otimes \left\{ p(c) \right\} \bigg)
  \\
  &=
  \textstyle \bigg( \left( \bigotimes_{c \in \mathcal{C}} \mathcal{P}_{\text{ZI MCAR}} (c) \right)
  \cap \left( \bigotimes_{c \in \mathcal{C}} \left\{ \mathbf{q}_{XRW \mid c} \mid \mathbf{q}_{W \mid R c} \mathbf{q}_{RX \mid c} = \mathbf{p}_{WX \mid c} \right\} \right) \bigg)
  \otimes \left\{ p(c) \right\}
  \\
  &=
  \textstyle \bigg( \bigotimes_{c \in \mathcal{C}} \left( \mathcal{P}_{\text{ZI MCAR}} (c)
  \cap \left\{ \mathbf{q}_{XRW \mid c} \mid \mathbf{q}_{W \mid R c} \mathbf{q}_{RX \mid c} = \mathbf{p}_{WX \mid c} \right\} \right) \bigg)
  \otimes \left\{ p(c) \right\}
  \\
  &=
  \left(\textstyle \bigotimes_{c \in \mathcal{C}}  \mathcal{Q}_{\text{ZI MCAR}}(c) \right) \otimes \left\{ p(c) \right\},
\end{aligned}
\end{equation}
where
\begin{equation}
\begin{aligned}
  \mathcal{Q}_{\text{ZI MCAR}}(c)
  &=
  \left\{ q(R, X, W \mid c) \mid {\bf Z, A1, A2}, \: \mathbf{q}_{W \mid R c} \mathbf{q}_{RX \mid c} = \mathbf{p}_{WX \mid c} \right\}.
\end{aligned}
\end{equation}
This is exactly the set $\mathcal{Q}$ described in Theorem~\ref{thm:zi_mcar_bound}.
Therefore, this equation suggests the application of Theorem~\ref{thm:zi_mcar_bound} to each stratum $C=c$.
First, there are marginal constraints: $\forall c, \forall x \neq 0, p_{w_0 \mid x, c} = p_{w_0 \mid x_1, c}$.
Second,
\begin{equation}
\begin{aligned}
  \mathcal{Q}_{\text{ZI MCAR}}(c)
  &=
  \left\{
      (\mathbf{q}_{W \mid R, c}, \mathbf{q}_{RX \mid c}) \mid
      \mathbf{q}_{RX \mid c} = [\mathbf{q}_{W \mid R \mid c}]^{-1} \mathbf{p}_{WX \mid c}, \mathbf{q}_{W \mid R} \in \mathcal{B},
  \right\}
\end{aligned}
\end{equation}
where $\mathcal{B}$ contains stochastic matrix $\mathbf{q}_{W \mid R c}$ satisfying
\begin{align*}
    q_{w_0 \mid r_1, c} &= p_{w_0 \mid x_1, c}\\
    q_{w_0 \mid r_0, c} &\in
    \begin{cases}
    [ p_{w_0 \mid x_0, c}, 1] \text{ if }p_{w_0 \mid x_0, c} > p_{w_0 \mid x_1, c}\\
    [ 0, p_{w_0 \mid x_0, c}] \text{ if }p_{w_0 \mid x_0, c} < p_{w_0 \mid x_1, c}\\
    (0, 1) \setminus \{ p_{w_0 \mid x_0, c} \}\text{ if } 0 < p_{w_0 \mid x_0, c} = p_{w_0 \mid x_1, c} < 1.
    \end{cases}
\end{align*}
Moreover, if $p_{w_0 \mid x_0, c} = p_{w_0 \mid x_1, c}$ then $0 < p_{w_0 \mid x_0, c} < 1$ is an additional condition,
and zero inflation does not occur for stratum $C=c$, i.e., $q(R=0 \mid c)=0$.
Since the compatibility set $\mathcal{Q}$ in this case is a Cartesian product of compatibility sets described in Theorem~\ref{thm:zi_mcar_bound}, which is sharp, $\mathcal{Q}$ is sharp.
\end{proof}

\pagebreak

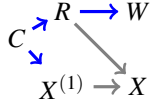
\begin{figure}[!htb]
  \centering
  \begin{tikzpicture}
    \tikzstyle{obs} = [circle, inner sep=1pt]
    \tikzstyle{unobs} = [circle, inner sep=0pt]
    \begin{scope}
      \path[->, very thick]
      node[unobs, inner sep=1pt] (x1) {$X^{(1)}$}
      node[obs, right of=x1, inner sep=1pt] (x) {$X$}
      node[unobs, above of=x1, inner sep=1pt] (r) {$R$}
      node[obs, below of=r, xshift=-0.6cm, yshift=+0.65cm, inner sep=1pt] (c) {$C$}
      node[obs, right of=r, inner sep=1pt] (w) {$W$}

      (x1) edge[gray] (x)
      (r) edge[gray] (x)
      (r) edge[blue] (w)
      (c) edge[blue] (r)
      (c) edge[blue] (x1)
      ;
    \end{scope}      
  \end{tikzpicture}
  \caption{
    The graph considered in Theorem~\ref{thm:zi_mar_bound_2}: proxy-augmented ZI MAR model satisfying \textbf{A1}$^*$ and \textbf{A2}$^*$ (Fig.~\ref{fig:zi_mcar_mar} (c) in the main paper).
  }
  \label{fig:zi_mar_apx_2}
\end{figure}

\begin{thma}{\ref{thm:zi_mar_bound_2}}
  Consider a ZI MAR model in Fig.~\ref{fig:zi_mcar_mar} (c) (reproduced in Fig.~\ref{fig:zi_mar_apx_2})
  under proxy assumptions {\bf A1$^*$} and {\bf A2$^*$},
  with categorical $X, C$ and binary $R, W$.
  Given a consistent observed law $p(X, W, C)$ satisfying positivity assumption, $\forall x, c, p(x,c) > 0$,
  the set of compatible proxy-indicator conditional distributions $q(W \mid R)$ is given by
  \begin{align*}
      q_{w_0 \mid r_1} &= p_{w_0 \mid x_1}\\
      q_{w_0 \mid r_0} &\in
      \begin{cases}
      [ \max_{c} p_{w_0 \mid x_0, c}, 1] \text{ if } \exists \tilde{c}, p_{w_0 \mid x_0, \tilde{c}} > p_{w_0 \mid x_1}\\
      [ 0, \min_c p_{w_0 \mid x_0, c}] \text{ if } \exists \tilde{c}, p_{w_0 \mid x_0, \tilde{c}} < p_{w_0 \mid x_1}\\
      (0, 1) \setminus \{ p_{w_0 \mid x_1} \}\text{ if } \forall c, p_{w_0 \mid x_0, c} = p_{w_0 \mid x_1}
      \end{cases}
  \end{align*}
  These bounds are sharp.
  Moreover, if $\forall c, p_{w_0 \mid x_0, c} = p_{w_0 \mid x_1}$, $p(X, W, C)$ must satisfy $\forall c, 0 < p_{w_0 \mid x_0, c} < 1$,
  and zero inflation does not occur, i.e., $q(R=0)=0$.
\end{thma}

\begin{proof}
\item
\textbf{Model definition.}
We assume $C$ is a cardinal variable, taking values in a finite set $\mathcal{C}$.
Any joint distribution in this ZI MAR model is
\begin{equation}
q \left( X^{(1)}, R, X, W, C \right) = q \left( X^{(1)}, R, X \mid C \right) q(W \mid R) q(C)
\end{equation}
Since the Markov factors are variationally independent, the ZI MAR model is a Cartesian product
\begin{equation}
\begin{aligned}
\mathcal{P}^{(1)}_{\text{ZI MAR}}
&=
\left( \textstyle \bigotimes_{c \in \mathcal{C}} \mathcal{P}_{X^{(1)}XR} (c) \right) \otimes \mathcal{P}_{W \mid R} \otimes \mathcal{P}_{C}
\\
\mathcal{P}_{C} &=\{q(C)\},
\\
\mathcal{P}_{W \mid R}  &=\{q( W \mid R) \mid \det \mathbf{q}_{W \mid R} \neq 0\},
\\
\mathcal{P}_{X^{(1)}XR} (c) &= \left\{q(X^{(1)}, R, X \mid c) \right\} \\
\end{aligned}
\end{equation}
Lemma~\ref{lm:zi_marginal_constraints} says $\mathcal{P}_{X^{(1)}XR} (c)$ is 1-to-1 to the set $\mathcal{P}_{XR}(c) = \left\{p(R, X \mid c) \mid {\bf Z}\right\}$.
Hence, we are interested in the marginal model
\begin{equation}
\begin{aligned}
\mathcal{P}_{\text{ZI MAR}}
&=
\left( \textstyle \bigotimes_{c \in \mathcal{C}} \mathcal{P}_{XR} (c) \right) \otimes \mathcal{P}_{W \mid R} \otimes \mathcal{P}_{C}.
\end{aligned}
\end{equation}

\textbf{Finding compatible set.}

Given observed law $\mathbf{p}_{WX C}$, we want to find the compatible set w.r.t. this law
\begin{equation}
\begin{aligned}
\mathcal{Q}
&= \left\{ \mathbf{q}_{XRWC} \mid \mathbf{q}_{XRWC} \in \mathcal{P}_{\text{ZI MAR}}, \: \forall c \in \mathcal{C} \left( \mathbf{q}_{W \mid R} \mathbf{q}_{RX \mid c} = \mathbf{p}_{WX \mid c} \right), \: \forall c \in \mathcal{C} \left( q(c) = p(c) \right) \right\}. \\
\end{aligned}
\end{equation}

This is similar to the compatible set we consider when {\bf A1$^\dag$, A2$^\dag$} hold (e.g., when $C \rightarrow W$), except the same $\mathbf{q}_{W \mid R}$ is shared between the constraints $\mathbf{q}_{W \mid R} \mathbf{q}_{RX \mid c} = \mathbf{p}_{WX \mid c}$. Each constraint restricts $\mathbf{q}_{W \mid R}$ in a different way, hence we cannot write $\mathcal{Q}$ as a Cartesian product to separate the constraints as we did before.

To proceed, note that $\mathcal{Q}$ is 1-to-1 to a set containing only $\mathbf{q}_{W \mid R}$, just as in ZI MCAR proof.
\begin{equation}
\begin{aligned}
\mathcal{Q}
&=
\left\{
    (\mathbf{q}_{W \mid R}, (\mathbf{q}_{RX \mid c})_{c \in \mathcal{C}}) \mid
    \forall c \left( \mathbf{q}_{RX \mid c} = [\mathbf{q}_{W \mid R}]^{-1} \mathbf{p}_{WX \mid c} \right), \: \mathbf{q}_{W \mid R} \in \mathcal{B}
\right\}
\otimes \{p(C)\}
\\
\mathcal{B}
  &=
  \left\{
    \begin{array}{c|c}
      \mathbf{q}_{W \mid R}
      &
      \begin{aligned}
        & & &\textstyle \mathbf{q}_{W \mid R} \geq 0, \: \forall r \left(\sum_w q_{w \mid r} = 1 \right), \: q_{w_0 \mid r_0} \neq q_{w_0 \mid r_1}, \\
        & & &\textstyle \text{for each } c: \mathbf{q}_{RX \mid c} \geq 0, \: \sum_{rx} q_{rx \mid c} = 1, \: \forall x \neq 0 (q_{r_0 x c} = 0), \\
        & & &\textstyle \quad \quad \quad \quad \quad  \text{where } \mathbf{q}_{RX \mid c} = [\mathbf{q}_{W \mid R}]^{-1} \mathbf{p}_{WX \mid c}.
      \end{aligned}
    \end{array}
  \right\}.
\end{aligned}
\end{equation}
This set is the intersection $\mathcal{B} = \cap_{c \in \mathcal{C}} \mathcal{B}_{c}$, in which each $\mathcal{B}_{c}$ contains only constraints associated with values $c$.
\begin{equation}
\begin{aligned}
\mathcal{B}_{c}
  &=
  \left\{
    \begin{array}{c|c}
      \mathbf{q}_{W \mid R}
      &
      \begin{aligned}
        & & &\textstyle \mathbf{q}_{W \mid R} \geq 0, \: \forall r \left(\sum_w q_{w \mid r} = 1 \right), \: q_{w_0 \mid r_0} \neq q_{w_0 \mid r_1}, \\
        & & &\textstyle \mathbf{q}_{RX \mid c} \geq 0, \: \sum_{rx} q_{rx \mid c} = 1, \: \forall x \neq 0 \left( q_{r_0 x c} = 0 \right), \\
        & & &\textstyle \text{where } \mathbf{q}_{RX \mid c} = [\mathbf{q}_{W \mid R}]^{-1} \mathbf{p}_{WX \mid c}. \\
      \end{aligned}
    \end{array}
  \right\}.
\end{aligned}
\end{equation}

We have already solved $\mathcal{B}_{c}$ before, it is the ZI MCAR compatibility set of $q(W \mid R)$ in Theorem~\ref{thm:zi_mcar_bound}. Then all we need is to take the intersection of these results, one for each $c$.
This intersection $\mathcal{B} = \cap_{c \in \mathcal{C}} \mathcal{B}_{c}$ is non-empty, because there is some $q(R, X, W, C)$ produces the given observed law.
First, the identification of $q_{w_0 \mid r_1}$ and marginal constraints are
\begin{equation}
  \label{eq:zi_mar_obs_constraint_in_proof}
  \forall c \in \mathcal{C}, \forall x \neq 0, q_{w_0 \mid r_1} = p_{w_0 \mid x_1, c} = p_{w_0 \mid x, c}.
\end{equation}
The last equality is due to the marginal constraint discussed in Theorem~\ref{thm:zi_mcar_bound}. Then we can write
\begin{equation}
  q_{w_0 \mid r_1} = p_{w_0 \mid x_1}.
\end{equation}

Next, we consider each case of the bound for $q_{w_0 \mid r_0}$.
\begin{enumerate}
  \item Suppose $p_{w_0 \mid x_0, c'} > p_{w_0 \mid x_1, c'}$ for some $c'$,
    then by Theorem~\ref{thm:zi_mcar_bound}, $q_{w_0 \mid x_0} > p_{w_0 \mid x_1, c'} = q_{w_0 \mid r_1}$,
    where last equality follows from equation~\ref{eq:zi_mar_obs_constraint_in_proof}.
  \item Suppose $p_{w_0 \mid x_0, c''} < p_{w_0 \mid x_1, c''}$ for some $c''$,
    then by Theorem~\ref{thm:zi_mcar_bound}, $q_{w_0 \mid x_0} < p_{w_0 \mid x_1, c''} = q_{w_0 \mid r_1}$,
    where last equality follows from equation~\ref{eq:zi_mar_obs_constraint_in_proof}.
\end{enumerate}

This means these 2 cases disjoint, i.e., we must have the following marginal constraint
\begin{equation}
    \text{either } \forall c \left(p_{w_0 \mid x_0, c} \leq p_{w_0 \mid x_1} \right) \text{ or } \forall c \left(p_{w_0 \mid x_0, c} \geq p_{w_0 \mid x_1}\right).
\end{equation}
The corresponding bounds are
\begin{align*}
    q_{w_0 \mid r_1} &= p_{w_0 \mid x_1}\\
    q_{w_0 \mid r_0} &\in
    \begin{cases}
    [ \max_{c} p_{w_0 \mid x_0, c}, 1] \text{ if } \exists c', p_{w_0 \mid x_0, c'} > p_{w_0 \mid x_1}, \\
    [ 0, \min_c p_{w_0 \mid x_0, c}] \text{ if } \exists c', p_{w_0 \mid x_0, c'} < p_{w_0 \mid x_1}, \\
    (0, 1) \setminus \{ p_{w_0 \mid x_1} \}\text{ if } \forall c, 0 < p_{w_0 \mid x_0, c} = p_{w_0 \mid x_1} < 1.
    \end{cases}
\end{align*}
The max/min appears since we take the intersection of the bounds for $c$.
Moreover, if $\forall c, p_{w_0 \mid x_0, c} = p_{w_0 \mid x_1}$, then $\forall c, 0 < p_{w_0 \mid x_0, c} < 1$ is an additional condition,
and zero inflation does not occur, i.e., $q(R=0) = 0$.
Due to the marginal constraints above, this exhausts all the cases.

Since the compatibility set $\mathcal{B}$ is the intersection of each $\mathcal{B}_c$, each is sharp in their own ZI MCAR model, the above abound is sharp.

\end{proof}

We collect the marginal constraints obtained from the proofs of Theorem~\ref{thm:zi_mar_bound_1} and Theorem~\ref{thm:zi_mar_bound_2} into the following lemma,

\begin{lma}{\ref{lm:zi_mar_obs_constraint}}
    For a ZI MAR model in Fig.~\ref{fig:zi_mcar_mar} (b) under 
    {\bf A1$^\dag$} and {\bf A2$^\dag$}, the observed law $p(X, W, C)$ obeys
    \begin{equation}
        \forall c, \forall x \neq 0, p_{w_0 \mid x, c} = p_{w_0 \mid x_1, c}.
    \end{equation}
    For a ZI MAR model in Fig.~\ref{fig:zi_mcar_mar} (c) under 
    {\bf A1$^*$} and {\bf A2$^*$}, the observed law $p(X, W, C)$ obeys
    \begin{align}
        & \forall c, \forall x \neq 0, p_{w_0 \mid x, c} = p_{w_0 \mid x_1}, \\
        \notag
        & \text{either } \forall c \left(p_{w_0 \mid x_0, c} \leq p_{w_0 \mid x_1} \right) \text{ or } \forall c \left(p_{w_0 \mid x_0, c} \geq p_{w_0 \mid x_1}\right).
    \end{align}
\end{lma}


\subsection{ZI MNAR proofs}

\begin{figure}[!htb]
  \centering
    \begin{tikzpicture}[>=stealth, node distance=1.35cm]
      \tikzstyle{obs} = [circle, inner sep=1pt]
      \tikzstyle{unobs} = [circle, inner sep=0pt]
      \begin{scope}
        \path[->, very thick]
        node[unobs, inner sep=1pt] (x11) {$X^{(1)}_1$}
        node[unobs, right of=x11, inner sep=1pt] (x21) {$X^{(1)}_2$}
        node[unobs, below of=x11, inner sep=1pt] (r1) {$R_1$}
        node[unobs, below of=x21, inner sep=1pt] (r2) {$R_2$}
        node[obs, above of=r1, xshift=-0.9546cm, inner sep=1pt] (x1) {$X_1$}
        node[obs, above of=r2, xshift=+0.9546cm, inner sep=1pt] (x2) {$X_2$}
        node[obs, below right of=r2, yshift=0.9546cm, inner sep=1pt] (w2) {$W_2$}
        node[obs, below left of=r1, yshift=0.9546cm, inner sep=1pt] (w1) {$W_1$}

        (x11) edge[blue] (x21)

        (r1) edge[blue] (w1)
        (r2) edge[blue] (w2)
        
        (x11) edge[blue] (r2)
        (x21) edge[blue] (r1)
        
        (x11) edge[gray, bend right=0] (x1)
        (r1) edge[gray] (x1)
        (x21) edge[gray, bend left=0] (x2)
        (r2) edge[gray] (x2)

        node[below of=r1, xshift=+0.65cm, yshift=+0.7cm] (l) {$(a)$}
        ;
      \end{scope}
      \begin{scope}[xshift=5cm]
        \path[->, very thick]
        node[unobs, inner sep=1pt] (x11) {}
        node[unobs, right of=x11, inner sep=1pt] (x21) {}
        node[unobs, below of=x11, inner sep=1pt] (r1) {$R_1$}
        node[unobs, below of=x21, inner sep=1pt] (r2) {}
        node[obs, above of=r1, xshift=-0.9546cm, inner sep=1pt] (x1) {$X_1$}
        node[obs, above of=r2, xshift=+0.9546cm, inner sep=1pt] (x2) {$X_2$}
        node[obs, below right of=r2, yshift=0.9546cm, inner sep=1pt] (w2) {$W_2$}
        node[obs, below left of=r1, yshift=0.9546cm, inner sep=1pt] (w1) {$W_1$}

        (r1) edge[blue] (w1)
        (x2) edge[blue] (x1)
        (x2) edge[blue] (r1)
        (x2) edge[blue] (w2)
        (w2) edge[blue] (x1)
        (w2) edge[blue] (r1)
        
        (r1) edge[blue] (x1)

        node[below of=r1, xshift=+0.65cm, yshift=+0.7cm] (l) {$(b)$}
        ;
      \end{scope}
    \end{tikzpicture}
  \caption{
    (a) The ZI bivariate block-parallel model.
    (b) The model Markov to this graph contains the marginal model for $p(X_1, X_2, R_1, W_1, W_2)$ in (a).
  }
  \label{fig:zi_mnar_apx}
\end{figure}
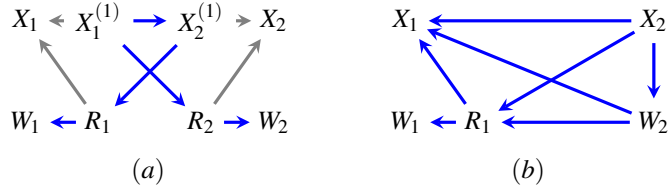

\begin{lma}{\ref{lm:general_obs_constraint}}
    Consider any ZI model in Section~\ref{subsec:zi_law_rest} under {\bf A1$^*$} and {\bf A2$^*$}. Denote $Z_k \triangleq \{X, W, C\} \setminus \{W_k, X_k\}$. The observed law $p(X, W, C)$ must satisfy, for each $k$,
    \begin{equation}
        \begin{cases}
        \forall z_k, \forall x \neq 0, p_{w_{k0} \mid x_{k}=x, z_k} = p_{w_{k0} \mid z_{k1}}, \\
        \forall z_k \left( p_{w_{k0} \mid x_{k0}, z_k} \leq p_{w_{k0} \mid x_{k1}} \right) \text{ or } \forall z_k \left(p_{w_{k0} \mid x_{k0}, z_k} \geq p_{w_{k0} \mid z_{k1}}\right).
        \end{cases}
    \end{equation}
\end{lma}

\begin{proof}
    Let $\tilde{\mathcal{G}}$ be the graph where $\{X_k, Z_k, R_k\}$ is fully connected, and $R_k \rightarrow W_k$. The Markov model $\mathcal{P}_{\tilde{\mathcal{G}}}$ for this graph contains all joint distributions $q(X_k, W_k, Z_k, R_k) = q(W_k \mid R_k) q(X_k, R_k, Z_k)$, where $q(X_k, R_k, Z_k)$ is from the saturated model restricted by \textbf{Z}. In the original ZI model, we have $p(X, W, C, R_k) = p(W_k \mid R_k) p(X_k, R_k, Z_k)$ with $p(X_k, R_k, Z_k)$ satisfying \textbf{Z}. Hence the model for this joint distribution is contained in $\mathcal{P}_{\tilde{\mathcal{G}}}$.

    Subsequently, one could repeat the proof of Theorem~\ref{thm:zi_mar_bound_2} to find the bound for $q(W \mid R)$, with $W_k$ as $W$, $X_k$ as $X$, and $Z_k$ as $C$. The bound is not sharp, because we do not know the exact model for $q(X_k, R_k, Z_k)$. Similarly, application of lemma~\ref{lm:zi_mar_obs_constraint} yields the desired constraints.
\end{proof}


\section{Simulations}

\subsection{Bound validity in random DGPs}

DGPs for ZI MCAR and ZI MAR are randomly selected according to Fig.~\ref{fig:zi_mcar_mar} (a) and (b), respectively.
In particular, a DGP for ZI MAR is a joint distribution which factorizes as
\begin{equation}
\begin{aligned}
p(X,X^{(1)}, R,W,C)
=
p(X \mid R, X^{(1)}) p(X^{(1)} \mid C) p(R \mid C) p(C) p(W \mid R)
\end{aligned}
\end{equation}
Then, the observed law is $p(X, W, C) = \sum_{X^{(1)}, R} p(X,X^{(1)}, R,W,C)$.

We randomly select a DGP by sampling the following parameters
\begin{equation}
    \begin{aligned}
        p(C=0) &\sim \operatorname{Uniform}[0,1] \\
        p(X^{(1)}=0 \mid C=0) &\sim \operatorname{Uniform}[0,1] \\
        p(X^{(1)}=0 \mid C=1) &\sim \operatorname{Uniform}[0,1] \\
        p(R=0 \mid C=0) &\sim \operatorname{Uniform}[0,1] \\
        p(R=0 \mid C=1) &\sim \operatorname{Uniform}[0,1] \\
        p(W=0 \mid R=0) &\sim \operatorname{Uniform}[0,1] \\
        p(W=0 \mid R=1) &\sim \operatorname{Uniform}[0,1] \\
    \end{aligned}
\end{equation}
Further more, to satisfy the ZI-consistency
\begin{equation}
    \begin{aligned}
        p(X=0 \mid R=0, X^{(1)}) &= 1; &  p(X=1 \mid R=1, X^{(1)}=1) &= 1; &  p(X=0 \mid R=1, X^{(1)}=0) &= 1.
    \end{aligned}
\end{equation}



\newpage
\subsection{Numerical bounds results}

We compute numerical bounds using method in \citet{duarte23automated} and compare to our analytical bounds for DGPs in ZI MCAR and ZI MAR.
Since computation time for the dual bound may be very long (some DGP might take more than 36 hours), we report only DGPs where primary bound is available (whose computation time may take only a few minutes). We refer reader to original paper for distinction of dual/primal bounds.

\begin{table}[ht]
    \footnotesize
    \centering
    \begin{tabular}{|c|c|c|c|c|c|}
    \hline
    \text{dgp}  &  \text{lb}  &  \text{ub}  &  \text{num lb}  &  \text{num ub}  &  $p_{w0 \mid r0}$ \\
    \hline \hline
0   &  0.556406  &  1.0       &  0.556411  &  1.0       &  0.820732  \\
1   &  0.357830  &  1.0       &  0.357830  &  1.0       &  0.493695  \\
2   &  0.0       &  0.520689  &  0.0       &  0.520689  &  0.453609  \\
4   &  0.606499  &  1.0       &  0.606499  &  1.0       &  0.682699  \\
5   &  0.0       &  0.524069  &  0.0       &  0.524061  &  0.496676  \\
6   &  0.381825  &  1.0       &  0.381825  &  1.0       &  0.441227  \\
8   &  0.652288  &  1.0       &  0.652288  &  1.0       &  0.659347  \\
9   &  0.698149  &  1.0       &  0.698149  &  1.0       &  0.738794  \\
10  &  0.0       &  0.443595  &  0.0       &  0.443595  &  0.442502  \\
11  &  0.656867  &  1.0       &  0.656867  &  1.0       &  0.850498  \\
12  &  0.211359  &  1.0       &  0.211359  &  1.0       &  0.856658  \\
14  &  0.183034  &  1.0       &  0.183034  &  1.0       &  0.303129  \\
15  &  0.648430  &  1.0       &  0.648430  &  1.0       &  0.833933  \\
16  &  0.292337  &  1.0       &  0.292337  &  1.0       &  0.307559  \\
17  &  0.500542  &  1.0       &  0.500542  &  1.0       &  0.553972  \\
18  &  0.0       &  0.102988  &  0.0       &  0.102988  &  0.087253  \\
20  &  0.0       &  0.479532  &  0.0       &  0.479532  &  0.238318  \\
21  &  0.426615  &  1.0       &  0.426615  &  1.0       &  0.426787  \\
22  &  0.399169  &  1.0       &  0.399169  &  1.0       &  0.494816  \\
23  &  0.0       &  0.216052  &  0.0       &  0.216052  &  0.158163  \\
24  &  0.436636  &  1.0       &  0.436636  &  1.0       &  0.533412  \\
26  &  0.429579  &  1.0       &  0.429579  &  1.0       &  0.710488  \\
27  &  0.0       &  0.500198  &  0.0       &  0.500199  &  0.451856  \\
28  &  0.0       &  0.383471  &  0.0       &  0.383471  &  0.136093  \\
29  &  0.0       &  0.325871  &  0.0       &  0.325871  &  0.070747  \\
30  &  0.363744  &  1.0       &  0.363744  &  1.0       &  0.374293  \\
    \hline
    \end{tabular}
    \caption{
        Comparison between our analytical lower and upper bound (\textit{lb}/\textit{ub})
        to numerical bounds (\textit{num lb}/\textit{num ub}) for a randomly selected set of DGPs in ZI MCAR model corresponding to Fig.~\ref{fig:zi_mcar_mar} (a)
        (reproduced in Fig.~\ref{fig:zi_mcar_apx}). True $p_{w_0|r_0}$ is reported.
    }
    \label{tab:zi_mcar_compare}
\end{table}

\newpage

\begin{table}[ht]
    \footnotesize
    \centering
    \begin{tabular}{|c|c|c|c|c|c|}
    \hline
    \text{dgp}  &  \text{lb}  &  \text{ub}  &  \text{num lb}  &  \text{num ub}  &  $p_{w0 \mid r0}$ \\
    \hline \hline
0   &  0.0       &  0.429089  &  0.0       &  0.429089  &  0.413267  \\
1   &  0.834644  &  1.0       &  0.834644  &  1.0       &  0.848638  \\
2   &  0.0       &  0.340484  &  0.0       &  0.340484  &  0.319264  \\
3   &  0.300217  &  1.0       &  0.300217  &  1.0       &  0.515513  \\
4   &  0.582249  &  1.0       &  0.582249  &  1.0       &  0.688620  \\
5   &  0.938604  &  1.0       &  0.938604  &  1.0       &  0.991572  \\
6   &  0.0       &  0.147758  &  0.0       &  0.147758  &  0.053637  \\
7   &  0.534321  &  1.0       &  0.534321  &  1.0       &  0.569545  \\
8   &  0.720775  &  1.0       &  0.720775  &  1.0       &  0.726467  \\
9   &  0.585611  &  1.0       &  0.592962  &  1.0       &  0.686385  \\
10  &  0.261442  &  1.0       &  0.261442  &  1.0       &  0.303129  \\
11  &  0.378136  &  1.0       &  0.378136  &  1.0       &  0.481036  \\
12  &  0.0       &  0.729282  &  0.0       &  0.729282  &  0.703234  \\
13  &  0.425249  &  1.0       &  0.425249  &  1.0       &  0.426797  \\
14  &  0.612665  &  1.0       &  0.612665  &  1.0       &  0.632688  \\
15  &  0.319180  &  1.0       &  0.319180  &  1.0       &  0.628988  \\
16  &  0.0       &  0.702582  &  0.0       &  0.702582  &  0.660187  \\
17  &  0.661849  &  1.0       &  0.661849  &  1.0       &  0.726963  \\
18  &  0.594456  &  1.0       &  0.594456  &  1.0       &  0.600720  \\
19  &  0.531541  &  1.0       &  0.532509  &  1.0       &  0.536156  \\
21  &  0.596110  &  1.0       &  0.600331  &  1.0       &  0.606306  \\
22  &  0.513144  &  1.0       &  0.513144  &  1.0       &  0.692834 \\
23  &  0.0       &  0.560519  &  0.0       &  0.560519  &  0.536384  \\
24  &  0.837194  &  1.0       &  0.837212  &  1.0       &  0.844800  \\
25  &  0.0       &  0.443658  &  0.0       &  0.443658  &  0.302563  \\
26  &  0.469323  &  1.0       &  0.469323  &  1.0       &  0.479800  \\
28  &  0.688084  &  1.0       &  0.688084  &  1.0       &  0.826820  \\
30  &  0.0       &  0.230720  &  0.0       &  0.230720  &  0.103594  \\
    \hline
    \end{tabular}
    \caption{
        Comparison between our analytical lower and upper bound (\textit{lb}/\textit{ub})
        to numerical bounds (\textit{num lb}/\textit{num ub}) for a randomly selected set of DGPs in ZI MAR model corresponding to Fig.~\ref{fig:zi_mcar_mar} (c)
        (reproduced in Fig.~\ref{fig:zi_mar_apx_2}). True $p_{w_0|r_0}$ is reported.
    }
    \label{tab:zi_mar_compare}
\end{table}


\newpage

\section{Variable descriptions in the CLABSI Data Application}

This section describes the covariates used in the CLABSI data application (all coded as binary variables).  These covariates correspond to types of therapy, and types of catheter used.
\begin{itemize}
 \item \texttt{Pediatrics}: the CVC therapy is tailored for children.
 \item \texttt{Chemotherapy}: the CVC therapy is used to administer chemotherapy.
 \item \texttt{OPAT}: outpatient parenteral antimicrobial therapy (IV antibiotics).
 \item \texttt{TPN}: parenteral nutrition delivered via the VC.
 \item \texttt{Other therapy}: any other type of therapy not included in the above categories, such as hydration.
 \item \texttt{Port}: a type of CVC in use.
 \item \texttt{PICC}: peripherally inserted central catheter, another type of CVC.
 \item \texttt{Tunneled CVC}: a CVC tunneled under the skin.
\end{itemize}

\end{document}